\documentclass[conference,12pt,onecolumn,draftclsnofoot]{IEEEtran}
\usepackage{geometry}
\usepackage{blindtext, graphicx}
\usepackage{amsmath,amsfonts,amssymb,bbm}
\usepackage{enumitem}
\usepackage{amsthm}

\DeclareMathOperator*{\argmin}{arg\,min}
\usepackage[T1]{fontenc}

\usepackage{epstopdf}
\usepackage{epsfig}
\graphicspath{{images/}}
\usepackage{dsfont}
\usepackage{psfrag,bbm,graphicx}
\usepackage{comment,balance}
\usepackage{hyperref}
\usepackage{epstopdf}
\usepackage{epsfig}
\usepackage{multicol}
\usepackage{multirow, cite}
\usepackage{amsfonts}
\usepackage{color}
\usepackage{colortbl}
\usepackage[ruled,vlined]{algorithm2e}
\hypersetup{
	colorlinks=true,
	linkcolor=blue,
	filecolor=magenta,      
	urlcolor=cyan,
}

\ifCLASSINFOpdf

\else

\fi

\newcommand{\QMONE}{{QM1}}
\newcommand{\QMTWO}{{QM2}}
\newcommand{\QMONENOISY}{{QM1-N}}
\newcommand{\QMTWONOISY}{{QM2-N}}

\newtheorem{theorem}{Theorem}
\newtheorem*{theorem*}{Theorem}
\newtheorem{lemma}{Lemma}
\newtheorem*{lemma*}{Lemma}

\newtheorem*{remark*}{Remark:}
\newtheorem{claim}[theorem]{Claim}
\newcommand{\remove}[1]{}
\allowdisplaybreaks

\begin{document}

\title{Query complexity of heavy hitter estimation
}

	\author{Sahasrajit Sarmasarkar, Kota Srinivas Reddy, and Nikhil Karamchandani \\
	Department of Electrical Engineering \\
	Indian Institute of Technology, Bombay \\
	Email: sahasrajit1998@gmail.com, ksreddy@ee.iitb.ac.in, nikhilk@ee.iitb.ac.in
}
 \maketitle
\begin{abstract}
We consider the problem of identifying the subset $\mathcal{S}^{\gamma}_{\mathcal{P}}$ of elements in the support of an underlying distribution $\mathcal{P}$ whose probability value is larger than a given threshold $\gamma$, by actively querying an oracle to gain information about a sequence $X_1, X_2, \ldots$ of $i.i.d.$ samples drawn from $\mathcal{P}$. We consider two query models: $(a)$ each query is an index $i$ and the oracle return the value $X_i$ and $(b)$ each query is a pair $(i,j)$ and the oracle gives a binary answer confirming if $X_i = X_j$ or not. For each of these query models, we design sequential estimation algorithms which at each round, either decide what query to send to the oracle depending on the entire history of responses, or decide to stop and output an estimate of $\mathcal{S}^{\gamma}_{\mathcal{P}}$, which is required to be correct with some pre-specified large probability. We provide upper bounds on the query complexity of the algorithms for any distribution $\mathcal{P}$ and also derive lower bounds on the optimal query complexity under the two query models. We also consider noisy versions of the two query models and propose robust estimators which can effectively counter the noise in the oracle responses.  
\end{abstract}

\maketitle
\section{INTRODUCTION}
Estimating the likely `heavy hitters' amongst the possible outcomes of an unknown probability distribution can be a useful primitive for several applications,  ranging from clustering and natural language processing to network flow / cache management  and online advertising. In this work, we formulate and study a Probably Approximately  Correct (PAC) sequential estimation problem in which the learner has access to a stream of independent and identically distributed (i.i.d.) samples from an underlying unknown distribution over a given support set via an oracle which it can query. The goal of the learner is to identify the set of elements in the support of the distribution whose probability value is above a pre-defined threshold. We will henceforth refer to this problem as \textit{threshold-based support identification}. We consider two natural models for oracle queries: \textit{(a) direct query} where the learner provides a sample index and the oracle responds with the value of the corresponding sample; and \textit{(b) pairwise query} where the learner queries the oracle with a pair of indices and the oracle responds with a binary answer confirming whether the sample values are identical or not. Note that in the pairwise query model, the true values of the samples are not revealed. While the former model has been a staple in the online learning literature \cite{devroye2002distribution}, the latter has also received significant attention recently under a wide variety of settings \cite{mazumdar2017clustering, chien2018query, jamieson2011active, mazumdar2016clustering,NIPS2016_6499}. The broad goal of our work is to design query-efficient schemes for these oracle models which can reliably estimate the support elements with probability values above a given threshold. 

A concrete application of the above setting (and a key motivation for the formulation) can be found in a clustering problem where we have an underlying collection of items which can be partitioned into a given number of clusters based on some inherent property, for example, products in a shopping platform based on category or a population of individuals based on their political preferences. However, unlike the usual setting where the goal is \textit{full clustering}, i.e., mapping each individual item to its corresponding cluster \cite{mazumdar2017clustering,mazumdar2017query}, we consider \textit{partial clustering} where the goal is relaxed to identifying cluster indices with size larger than a certain fraction of the population.  Now, say a stream of samples is generated from the population by picking an item uniformly at random in each instant. Then, it is easy to see that the probability of picking an item from a certain cluster is proportional to the size of the cluster and thus, the problem of identifying clusters with size greater than a fraction of the population would correspond to solving the threshold-based support identification problem under the corresponding sampling distribution as mentioned above. Our results provide schemes for performing partial clustering which can have significantly smaller query complexity than the naive approach of performing full clustering and then selecting the appropriate clusters. For example, say we have $n$ items with a ground truth clustering with $k = \Theta(\sqrt{n})$ clusters, where the three largest clusters have size $n^{7/8}$ and the remaining clusters are  of size $O(\sqrt{n})$. If we have to identify all the clusters with size more than $n^{3/4}$ with access to a pairwise oracle, Theorem~\ref{QM2ub} shows that this can be performed using $O(n)$ queries whereas the naive approach would require $O(nk) = O(n^{3/2})$ queries \cite{mazumdar2016clustering}.

We make the following contributions to understanding the query complexity for this class of sequential estimation problems. For both the query models, we first design sequential algorithms which can provably solve the threshold-based support identification problem with large probability and also provide upper bounds on their query complexity. The estimators are based on maintaining empirical estimates and confidence intervals for the probability values associated with each element of the support. We also provide information-theoretic lower bounds on the query complexity of any reliable estimator under these query models. The bounds presented are instance-specific, i.e., they depend on the underlying probability distribution and the value of the threshold. Finally, we also consider noisy versions of both the direct query as well as the pairwise query models, and propose robust estimators which can effectively counter the noise in the oracle responses.  
%
%
\subsection{Related Work}
Estimation of properties of  distributions using samples is a direction of work which has a long and rich history. While some of the older works were concerned with the statistical properties of the estimators such as consistency etc., see for example  \cite{chernoff1964estimation, parzen62estimation} which study mode estimation, there has been a lot of work recently on characterizing the optimal query complexity for estimating various properties of probability distributions including entropy \cite{caferov2015optimal, acharya2016estimating}, support size and coverage \cite{hao2019data, wu2018sample}, and `Lipschitz' properties \cite{hao2019unified} amongst others. Another related line of work is the `heavy hitter' estimation problem in the context of streaming algorithms \cite{sivaraman2017heavy, bhattacharyya2018optimal, karp2003simple} where given an \textit{arbitrary} stream of samples from a large alphabet, the goal is to identify those symbols whose frequency is above a certain threshold. The metric of performance is usually computational in nature such as the memory size, the number of passes or run time complexity. In contrast to these works, we study the threshold-based support identification problem in the \textit{stochastic} setting and our interest is in deriving instance-specific bounds on optimal query complexity which illustrate the dependence on the underlying distribution. Also, we study the query complexity for reliable estimation under several query models including the pairwise query model which isn't as prevalent in the literature. 


Online decision making and active learning are also key features of the popular framework of Multi-Armed Bandits (MABs) \cite{MABbook}. In particular, the problem of \textit{thresholding bandits} \cite{locatelli2016optimal} aims to find the set of arms whose mean reward is above a certain threshold, using the minimum number of arm pulls. Thinking of each element in the support set as an arm, the MAB problem is indeed quite related to the setting studied in this work, especially the one with pairwise queries and we explore this relationship in this paper. One key difference between the MAB problem and our setup is that in the former we can choose to pull any arm at each instance, whereas in our problem the arm to be pulled is determined exogenously by the samples generated by the underlying probability distribution.

The two lines of work closest to ours in spirit are \cite{mazumdar2017clustering,mazumdar2017query} and \cite{shah2020sequential, 8732224}. The goal in \cite{mazumdar2017clustering,mazumdar2017query} is to fully cluster a collection of items and they characterize the optimal query complexity\footnote{Unlike our work, these results are worst-case and not instance-specific} for this task. In this context, our work corresponds to the less stringent objective of identifying the `significantly large' clusters which might be a  natural objective in several unsupervised machine learning applications where we wish to associate labels to a large fraction of the unlabelled dataset and has not been studied before, to the best of our knowledge. On the other hand, \cite{shah2020sequential, 8732224} study the related problems of identifying the top-$m$ clusters and mode of an underlying distribution and find bounds on the optimal query complexity for these tasks. While the broad structure of our algorithms and lower bound techniques are similar, the details differ greatly, for example in the choice of confidence intervals in achievability and alternate instances for the converses. One additional technical challenge that we face due to our different objective is that the number of clusters to be recovered is not fixed apriori (all or one or $m$), but depends on the particular problem instance and this requires careful treatment in the algorithms as well as in the analysis.   

\section{PROBLEM FORMULATION} \label{sec: Problem}

Consider an unknown discrete probability distribution $\mathcal{P}=\{p_1,p_2,...,p_k\}$ over a finite support set $\{1,2,...,k\}$. A random variable $X$ is said to be sampled from a distribution $\mathcal{P}$, if $\mathbb{P}_{\mathcal{P}}(X=i) = p_i$ for all $i\in\{1,2,...,k\}$. Our goal in this paper is to 
identify all the indices whose probability of occurrence is above a given threshold $\gamma$. Formally, we want to identify the set of support elements $S_{\mathcal{P}}^{\gamma} = \{i \  | \ p_i \geq \gamma\}$. Towards this, we have access to an independent and identically distributed (i.i.d.) sequence of samples $X_1,X_2,...$, from $\mathcal{P}$ via an oracle. 
 We consider the following four types of oracle queries.
\begin{itemize}
	\item \textit{Noiseless query model 1 (\QMONE)}: Queried with an index $i$, the oracle response is given by $\mathcal{O}(i) = X_i$, i.e., the oracle returns the value $X_i$ of the $i^{th}$ sample.
	\item \textit{Noisy  query model 1 (\QMONENOISY)}: Queried with an index $i$, the oracle returns the value $X_i$ of the $i^{th}$ sample with probability $1-p_e$ and a random value in $\{1,2,...,k\}$ with probability $p_e$, for some $p_e \in (0, 1/2)$. More formally, the oracle response\footnote{Note that the oracle response remains the same if the same index $i$ is queried repeatedly.} to query $i$ is given by $\mathcal{O}(i)= (1- K_i)\times X_i+ K_i \times U_i$, where $K_i$ is a Bernoulli random variable with mean $p_e$ and $U_i$ is a uniform random variable over the set $\{1,2,\ldots,k\}$. We assume that $\{U_i\}$, $\{K_i\}$ are all independent random variables and are also independent of the sequence of samples $X_1, X_2, \ldots$. 
	\item \textit{Noiseless query model 2 (\QMTWO)}: In this model, the oracle makes pairwise comparisons. Given two indices $i$ and $j$, the oracle tells us whether the values $X_i$ and $X_j$ are equal or not.  Formally, we denote the oracle response to a query pair $(i,j)$ as:
	\begin{equation*}
	\mathcal {O}(i,j)=
	\begin{cases}
	1 & \text{if}\ X_i = X_j, \\
	-1 & \text{otherwise.}
	\end{cases}
	\end{equation*}
	Note that the true values of $X_i$ or $X_j$ are not revealed in this model.
		\item \textit{Noisy query model 2 (\QMTWONOISY)}: The oracle model here is the same as the one in \QMTWO, except that the true noiseless
		\QMTWO \ oracle response is flipped with probability $p_e$. In particular, given two indices $i$ and $j$, if the values $X_i$ and $X_j$ are equal, the oracle returns $+1$ with probability $1-p_e$ and $-1$ with probability $p_e$.  Similarly, if the values of $X_i$ and $X_j$ are not equal, the oracle answers $+1$ with probability $p_e$ and $-1$ with probability $1-p_e$. Formally, we denote the oracle response to a query pair $(i,j)$ as:
        $$\mathcal{O} (i, j) = (2\times \mathbbm{1}_{X_i = X_j}-1) \times (1 - 2Z_{i,j}),$$
  where $\mathbbm{1}_E$ denotes the indicator random variable corresponding to event $E$ and $Z_{i,j}$ denotes a Bernoulli random variable with mean $p_e$. We assume that $\{Z_{i,j}\}$ are independent random variables and are also independent of the sequence of samples $X_1, X_2, \ldots$.
\end{itemize}
For each of the query models described above, we aim to design efficient sequential algorithms which proceed in rounds as follows: 
\begin{itemize}
    \item The algorithm chooses a pair of indices (an index) for the \QMTWO  \textbackslash \QMTWONOISY  \ (\QMONE \textbackslash \QMONENOISY)  models and queries the oracle with those indices (index).
    \item Based on all the responses received from the oracle thus far, the algorithm decides either to terminate or proceed to the next round.
\end{itemize}
When the algorithm decides to terminate, it returns a set of indices denoted by $\widehat{S}$ as an estimate of the set $S_{\mathcal{P}}^{\gamma}$ consisting of all indices $i$ in the support set for which the corresponding probability value $p_i$ is above $\gamma$.     

We measure the cost of an estimator in terms of the number of queries made before it stops and its accuracy in terms of the probability with which it successfully estimates the index set $S_{\mathcal{P}}^{\gamma}$. For $0<\delta<1$ and given $0<\gamma<1$, an algorithm is defined to be a $\delta$-true $\gamma$-threshold estimator, if for every underlying distribution $\mathcal{P}$, it identifies the set $S_{\mathcal{P}}^{\gamma}$ with probability at least $1- \delta$, i.e., $\mathbb{P}_\mathcal{P}[\widehat{S}= S_{\mathcal{P}}^{\gamma}] \geq (1- \delta)$, where $\widehat{S}$ is the estimated support set. In this work, we aim to design efficient $\delta$-true $\gamma$-threshold estimators which require as few queries as possible. For a $\delta$-true $\gamma$-threshold estimator $\mathcal{A}$ and a distribution $\mathcal{P}$, let $Q_{\delta,\gamma}^{\mathcal{P}}(\mathcal{A})$ denote the number of queries made by the estimator. Note that $Q_{\delta,\gamma}^{\mathcal{P}}(\mathcal{A})$ is itself a random variable and our results hold either with high probability or in expectation.

Without loss of generality (W.L.O.G.), we assume $p_1 \geq p_2 \geq ...\geq p_m > \gamma> p_{m+1} \geq ...\geq p_k $, i.e., a $\delta$-true $\gamma$-threshold estimator must return the set $\{1,2,...,m\}$ with probability at least $1-\delta$. The rest of the paper is organized as follows. We propose sequential estimation algorithms for the \QMONE \ and \QMONENOISY \ models in Sections~\ref{Sec:QM1} and \ref{Sec:QM1N} respectively, where we provide both upper and lower bounds on the query complexity. Section~\ref{Sec:QM2} discusses the query complexity and our proposed algorithm for the \QMTWO \ model. Finally, Section~\ref{Sec:QM2N} proposes an estimator for the \QMTWONOISY \ model and analyses its query complexity. Some numerical results are provided in Section~\ref{Sec:Sim} and we relegate all the proofs to Section~\ref{Sec:Proofs}.
\section{\QMONE \ THRESHOLD ESTIMATOR}
\label{Sec:QM1}
Under \QMONE, for every query with index $i$, the oracle returns the value $X_i$ of the $i^{th}$ sample. We now present an algorithm for this query model and analyze its query complexity, see Algorithm~\ref{QM1_alg}.
\subsection{Algorithm}
We first create $k$ bins numbered $1,2,...,k$. In each time step $t$, we query the oracle with the index $t$. The oracle reveals the value $X_t$ and the index $t$ is placed in the bin whose index matches $X_t$. 

We define $Z_i^t$ for $i\in \{1,2,..,k\}$ used in Algorithm \ref{QM1_alg} as follows:
     \begin{equation*}
    Z_i^t=
    \begin{cases}
      1  & \text{if}\ X_t = i, \\
      0  & \text{otherwise}.
    \end{cases}
  \end{equation*}
We can argue that for each given $i$, $\{Z_i^t\}$ is a collection of i.i.d. Bernoulli random variables with $\mathbb{E}[Z_i^t]=p_i$. Let $\tilde{p}_i^t $ denote the empirical probability estimate for support element $i$ at time $t$, and is given by 
 \begin{equation}\label{eq1}
     \tilde{p}_i^t = \sum_{j=1}^{t} Z_i^j/t.
 \end{equation}
Note that  $\tilde{p}_i^t$  denotes the fraction of samples in Bin $i$, till time $t$.

 At each round $t$ and for every $i$,  we choose the confidence interval of Bin $i$ such that  $p_i$ lies within the interval with "sufficiently" high probability. For a given time $t$, let $u_i(t)$ and $l_i(t)$  denote the upper confidence bound (UCB) and lower confidence bound (LCB) of Bin $i$ respectively and the confidence interval of Bin $i$ is given by $[l_i(t),u_i(t)]$. 
 
{\begin{algorithm}{\label{QM1_alg}}
{	{
		\SetAlgoLined
		Create bins numbered $1,2...k$.\\
		Initialize $\tilde{p}_i^{0}=0$, $l_i(0)=0$, $u_i(0)=1$, $\forall i\in\{1,2,...,k\}$.\\
		$t = 0$ \\
		\While{$\exists$ Bin $j$ s.t $l_j(t)$ $<$ $\gamma$ $<$ $u_j(t)$}{
			{
				Obtain $X_{t+1}$.
			} 
			
			$t = t +1$.\\
			Place the index $t$ in the bin numbered $X_t$.\\
			\text{Update the empirical estimate} $\tilde{p}_i^{t}$ according to \eqref{eq1}; upper and lower confidence bounds $l_i(t)$ according to \eqref{eqnlb} and $u_i(t)$ \eqref{eqnub}  respectively for all bins.\\

			%
		}
		Output  $\{i|l_i(t)>\gamma\}$.
		\caption{Estimator for \QMONE}
	}}
\end{algorithm}}

Motivated from the  bounds in \cite{pmlr-v30-Kaufmann13}, for some given sequence of parameters $\{\beta^t\}$,   we define $l_i(t)$  and $u_i(t)$ as follows:
\begin{align}
l_i(t) &= {\min}\{ q \in [0,\tilde{p}_i^t]: t\times d(\tilde{p}_i^t||q) \leq \beta^t  \},\label{eqnlb} \\
u_i(t) &= {\max}\{ q \in [\tilde{p}_i^t,1]: t\times d(\tilde{p}_i^t||q) \leq \beta^t  \}.\label{eqnub}
\end{align}
Here $d(p||q)$ denotes the  Kullback-Leibler divergence between two Bernoulli distributions with means $p$ and $q$ respectively.


The algorithm keeps querying the oracle and updating the confidence intervals for each bin, until every bin has either its LCB above $\gamma$ or its UCB below $\gamma$ at which point the Algorithm~\ref{QM1_alg} terminates. The estimate $\widehat{S}$ is chosen to consist of all bin indices which have their LCBs above $\gamma$. We choose the sequence $\{\beta^t\}$ in such a way that the probability of error is bounded by $\delta$.
%
 
 The following lemma claims that Algorithm \ref{QM1_alg} returns the desired set $S_{\mathcal{P}}^{\gamma}$ for any underlying distribution $\mathcal{P}$ with probability at least $1- \delta$.

\begin{lemma}{\label{Correctness_QM1}}
Given the choice of $\beta^t = \log(2kt^2 / \delta)$ for each $t\ge 1$, Algorithm \ref{QM1_alg} is a $\delta$-true $\gamma$-threshold estimator under \QMONE. 
\end{lemma}
\subsection{Query Complexity Analysis}
Given two values $x$ and $y$ satisfying $0 \leq x,y \leq 1$, we define $d^{*} (x,y) = d(z||x)$ where $z$ satisfies the condition $d(z||x) = d(z||y)$. The quantity $d^{*} (x,y)$ represents the \textit{Chernoff information} between two Bernoulli distributions with means $x$ and $y$ respectively, and is a relevant quantity in hypothesis testing problems \cite[Section 11.8]{10.5555/1146355}. This definition plays a key role in the following theorem which provides an upper bound on the query complexity of our proposed estimator in Algorithm \ref{QM1_alg}.

\begin{theorem}{\label{QM1ub}}
	Let $\mathcal{A}$ denote the estimator in Algorithm  \ref{QM1_alg} with  $\beta^t = \log(2kt^2 / \delta)$ for each $t\ge 1$ and let $Q_{\delta,\gamma}^{\mathcal{P}}(\mathcal{A})$ be the corresponding query complexity for a given distribution $\mathcal{P}$ under \QMONE. Then,  we have
	$$Q_{\delta,\gamma}^{\mathcal{P}}(\mathcal{A})\leq \max_{j\in\{m,m+1\}} \Biggl\{\frac{2e\log \Bigl(\sqrt{\frac{2k}{\delta}}\frac{2}{d^{*}(p_j,\gamma)}\Bigr)}{(e-1)d^{*}(p_j,\gamma)} \Biggr\},$$
	with probability at least $1-\delta$.
\end{theorem}

This is essentially argued as follows. For each $i$, we identify a time $T_i$ such that with high probability, Bin $i$ would have been classified as being either above $\gamma$ or  below $\gamma$ by time $T_i$. We then show that $T_i$ is bounded by the expression given in Theorem \ref{QM1ub}, $\forall i \in\{1,2,...,k\}$.

 
 Next, we  give a lower bound of the number of queries $Q_{\delta,\gamma}^{\mathcal{P}}(\mathcal{A})$ for any $\delta$-true $\gamma$-threshold estimator $\mathcal{A}$ under \QMONE. 
 
 \begin{theorem}\label{QM1lb}
   For any $\delta$-true $\gamma$-threshold estimator  $\mathcal{A}$ under \QMONE, let $Q_{\delta,\gamma}^{\mathcal{P}}(\mathcal{A})$ be the query complexity. Then, we have
   $$\mathbb{E}[Q_{\delta,\gamma}^{\mathcal{P}}(\mathcal{A})]\geq \max_{j\in\{m,m+1\}}\Bigg\{\frac{\log{\frac{1}{2.4\delta}}}{d(p_j||\gamma)}\Bigg\}.$$
 
 \end{theorem}
 
The proof of the above result is based on change of measure arguments similar to those in \cite{kaufmann2016complexity}. The detailed proof is given in Section~\ref{Sec:Thm2_proof}. Theorems~\ref{QM1ub} and \ref{QM1lb} provide a fairly tight characterization of the optimal query complexity under the QM1 query model, with the key difference between the upper and lower bound being the terms $d^{*}(p_j,\gamma)$ and $d(p_j||\gamma)$ respectively. {From the definition of $d^*(x,y)$, $\min\bigg\{d\Big(\frac{p_j+\gamma}{2}||p_j\Big),d\Big(\frac{p_j+\gamma}{2}||\gamma\Big)\bigg\}\leq d^{*}(p_j,\gamma) \leq \max\bigg\{d\Big(\frac{p_j+\gamma}{2}||p_j\Big),d\Big(\frac{p_j+\gamma}{2}||\gamma\Big)\bigg\}$. We also have $ 2(p_j-\gamma)^2 \leq d(p_j||\gamma) \leq \frac{(p_j-\gamma)^2}{\gamma(1-\gamma)}$ (\cite{popescu2016bounds}), and combining the above inequalities, we get $\frac{d(p_j||\gamma)}{d^{*}(p_j,\gamma)}\leq \frac{2}{\gamma(1-\gamma)}$. Thus, our bounds in Theorems \ref{QM1ub} and \ref{QM1lb} are tight up to logarithmic factors.}

\begin{remark*}
	Our algorithms use the knowledge of the number of support elements ($k$) in the confidence intervals' design, whereas our lower bounds are independent of $k$. Whether the dependence on $k$ is fundamental or can it be removed by designing smarter algorithms and using more sophisticated concentration inequalities is an important question we want to address in the future.
\end{remark*}
%
 %
 \section{\QMONENOISY \  THRESHOLD ESTIMATOR }
 \label{Sec:QM1N}
 Under \QMONENOISY, when queried with an index $i$, the oracle returns the true value of $X_i$ with probability $1-p_e$ and a uniformly at random chosen value in $\{1,2,\ldots,k\}$ with probability $p_e$.
 	The responses of the oracle in this noisy setting would stochastically be the same as the responses from another oracle in the \QMONE \ model with an underlying distribution $p_i' = (1-p_e)p_i + {p_e}/{k}$ over the same support set. Since the responses from the oracles in the two settings described above are stochastically identical, Theorems \ref{QM1ub} and \ref{QM1lb} also provide upper and lower bounds on the query complexity of $\delta$-true $\gamma$-threshold estimators for the \QMONENOISY \ model, by considering the underlying distribution as $\mathcal{P}'=\{p_1',p_2',...,p_k'\}$, and the threshold as $\gamma'$ = $(1-p_e)\gamma + {p_e}/{k}$. Note that the estimator proposed in Algorithm \ref{QM1_alg} for the \QMONE \ model is a $\delta$-true $\gamma$-threshold estimator for the noisy setting as well, with the threshold set to $\gamma' = (1-p_e)\gamma + {p_e}/{k}$.
 \section{\QMTWO \ THRESHOLD-ESTIMATOR }
 \label{Sec:QM2}
Recall that under the  \QMTWO \ model, the oracle only makes pairwise comparisons. When queried with two indices $i$ and $j$, the oracle response indicates whether the values $X_i$ and $X_j$ are equal or not. We now present $\delta$-true $\gamma$-threshold estimator for this query model and analyse its query complexity.
 \subsection{Algorithm}
 \textit{High-level description:} Recall that under the \QMONE \ model, we begin by creating $k$ bins, one corresponding to each element of the support of the underlying distribution. In each time slot $t$, the oracle reveals the value $X_t$ of the sample $t$, and the index $t$ is placed in the corresponding bin. A key challenge in the \QMTWO \ model is that the oracle does not reveal the value of the samples. Instead, it just provides pairwise information about whether the values of two samples corresponding to the pair of queried indices are equal or not. Here, we create and update bins as the algorithm proceeds, while ensuring that each such bin contains samples representing the same support element.
As with the previous models, the algorithm proceeds in rounds. Let $\mathcal{B}(t)$ denote the set of bins created by the end of round $t$. For each round $t >1 $, we choose a subset of bins $\mathcal{C}(t-1)$ from $\mathcal{B}(t-1)$. We go over the bins in $\mathcal{C}(t-1)$ one by one, and then query the oracle to compare the $t^{th}$ sample with a representative element from the current bin.
Round $t$ stops when we either get a positive response from the oracle indicating that a matching bin has been found, or when all the bins in $\mathcal{C}(t-1)$ have been exhausted. As before, we maintain confidence intervals for the true probability value corresponding to each bin and our algorithm as described in Algorithm~\ref{QM2_alg} terminates, when its termination criterion is met, which we discuss below in the detailed description.

\textit{Detailed description:} 
We define $\mathbb{Z}_i^t$ as follows:
{\small
\begin{equation*}
    \mathbb{Z}_i^t = 
    \begin{cases}
      1  & \text{ if } \mathcal{O}(i_r,t)=1 \text{ for } i_r \in \text{Bin } i , \text{ Bin } i \in \mathcal{B}(t-1)\\
      0  & \text{otherwise}.
    \end{cases}
\end{equation*}}

Let $\hat{p}_i^t${\footnote{Note that $\hat{p}_i^t$ denotes the fraction of samples in Bin $i$ and $\tilde{p}_i^t$ in Section \ref{Sec:QM1} denotes the fraction of samples of the support element $i$.}} denotes the fraction of samples present in Bin $i$ by the end of round $t$. We define it formally as follows:
\begin{equation}{\label{emp_QM2}}
    \hat{p}_i^t = {\sum_{j=1}^{t} \mathbb{Z}_i^j }/{t}.
\end{equation}
\begin{algorithm}[h]{\label{QM2_alg}}
	{\small
		\SetAlgoLined
		$t$=1\\
		Create a new bin and add sample $t$ to it.\\
		Update the empirical estimate $\hat{p}_i^{t}$ according to \eqref{emp_QM2}, the upper bound $\hat{u}_i(t)$ and the lower bound  $\hat{l}_i(t)$ according to \eqref{eqnubbin} and \eqref{eqnlbbin} for the created bin.
		Form $\mathcal{C}$($t$) according to $\eqref{eqn:c1t}$.\\
		\While{( $t < T'$)}{
			{$t=t+1$\\}
			flag = 0.\\
			
			\ForAll{ $j$ $\in$ $\mathcal{C}(t-1)$ }{
				Obtain $j_l \in \text{Bin } j$.\\
				\If{$\mathcal{O}(t,j_l) == +1$}{
					
					Add sample $t$ to corresponding bin.\\
					flag = 1.	BREAK.\\
					
				}
			}
			
			\If{flag==0  } {
				
				Create a new bin and add sample $t$ to it.
			}
			
			Update the empirical estimate $\hat{p}_i^{t}$ according to \eqref{emp_QM2}, the upper bound $\hat{u}_i(t)$ and the lower bound  $\hat{l}_i(t)$ according to \eqref{eqnubbin} and \eqref{eqnlbbin} for all the  bins in $\mathcal{C}(t-1)$ and the newly created bin if created in this round.\\
			
			{For all other bins, we keep the empirical means and LCB-UCB bounds unchanged.\\} 
			
			%
			
			{Form $\mathcal{C}$($t$) according to $\eqref{eqn:c1t}$.}\\
		}
		
		Initialise $\mathcal{S}=\{x|\hat{l}_x(t)>\gamma\}$.\\
		\While{$\mathcal{C}(t)\neq \Phi$}{
			{ $t = t +1$}.\\
			
			\ForAll{$j$ $\in$ $\mathcal{C}(t-1)$ }{
				Obtain $j_l \in \text{Bin } j$.\\
				\If{$\mathcal{O}(t,j_l) == +1$}{
					
					Add sample $t$ to corresponding bin. 
					BREAK.\\
					
				}
			}
			
			%
			
			Update the empirical estimate $\hat{p}_i^{t}$ according to \eqref{emp_QM2}, the upper bound $\hat{u}_i(t)$ and the lower bound  $\hat{l}_i(t)$ according to \eqref{eqnubbin} and \eqref{eqnlbbin} for all the  bins in $\mathcal{C}(t-1)$.\\
			Update $\mathcal{S}$ by adding the bins with $\hat{l}_x(t)>\gamma$.
			%
			
			{Form $\mathcal{C}$($t$) according to $\eqref{eqn:c2t}$.}
		}
		Output  $\mathcal{S}$.
		\caption{Estimator for \QMTWO}
	}
\end{algorithm}
As before, let $\hat{l}_i(t)$ and $\hat{u}_i{(t)}$ denote the lower and upper confidence bounds for Bin $i$ respectively, and as before are defined as follows for some appropriate choice of $\{\beta^t\}$: 
\begin{align}
\label{eqnlbbin}
\hat{l}_i(t) & = {\min}\{ q \in [0,\hat{p}_i^t]: t\times d(\hat{p}_i^t||q) \leq \beta^t  \},\\
\label{eqnubbin}
\hat{u}_i(t) &= {\max}\{ q \in [\hat{p}_i^t,1]: t\times d(\hat{p}_i^t||q) \leq \beta^t  \}.
\end{align}
We run the algorithm in 2 phases. The goal of the first phase is to ensure that one bin is created corresponding to each support element $\{1,2,...,m\}$ in $S_{\mathcal{P}}^{\gamma}$ with "high" probability. The first phase runs from round $1$ to round ${T'=\big(\log (2k/\delta)\big)/{\big(\log\frac{1}{1-\gamma}}\big)}$. For this phase, the subset of bins which have their UCB above $\gamma$ form the subset $\mathcal{C}(t-1)$ of the created bins which are compared against in round $t$, i.e., 
%
%
\begin{align}\label{eqn:c1t}
\mathcal{C}(t)=\{x|  \gamma<\hat{u}_x(t) , x \in \mathcal{B}(t) \}.
\end{align}
%
%
{In each round $t$, we choose a representative index ${i_r}\in$ Bin $i$ , $\forall i\in\mathcal{C}(t-1)$ and then go over the bins in $\mathcal{C}(t-1)$ one by one, querying the oracle with index pairs of the form ($i_r,t$)}. If we get a positive response from a bin, we place the index $t$ in the corresponding bin. If the replies from all the bins in $\mathcal{C}(t-1)$ is $-1$, we create a new bin with index $t$. We update the empirical estimates $\{\hat{p}_i^t\}$ as well as the UCB and LCBs of the created bins appropriately. 

The second phase runs from round ${T'=\big(\log (2k/\delta)\big)/{\big(\log\frac{1}{1-\gamma}}\big)}$ onwards. In this phase, we do not create any new bins since the first phase guarantees that a bin corresponding to each element in $S_{\mathcal{P}}^{\gamma}$ has already been created. The goal of this phase is to correctly identify the bins belonging to elements in $S_{\mathcal{P}}^{\gamma}$ from amongst the created bins. Bins for which the corresponding LCB is greater than the threshold $\gamma$ are classified as belonging to  $S_{\mathcal{P}}^{\gamma}$; and vice versa for bins with UCB at most $\gamma$.

At any round $t$ in this phase, the bins still in contention are those for which the UCB is above $\gamma$ and the LCB is below $\gamma$, and these are the ones that are chosen to form $\mathcal{C}(t)$, i.e., 
\begin{equation}
\label{eqn:c2t}
\mathcal{C}(t)=\{x| \hat{l}_x(t)< \gamma<\hat{u}_x(t) , x \in \mathcal{C}(t-1) \}.
\end{equation}
%

Similar to the first phase, in round $t$ we go over the bins in $\mathcal{C}(t-1)$ one by one, and then query the oracle to compare the $t^{th}$ sample with a representative element from the current bin. However, unlike the first phase, note that in this phase if we get a negative response from all the bins in $\mathcal{C}(t-1)$, we simply drop the index $t$.

The algorithm terminates when $\mathcal{C}(t)$ is empty, i.e., all the  bins either have their LCB greater than $\gamma$ or their UCB below $\gamma$.
The following lemma claims that Algorithm \ref{QM2_alg} returns the desired set $S_{\mathcal{P}}^{\gamma}$ for any underlying distribution $\mathcal{P}$ under the \QMTWO \ model, with probability at least ($1- \delta$). 
%
%
%
%

%
\begin{lemma}{\label{Correctness_QM2}}
	Given the choice of $\beta^t = \log(4kt^2/\delta)$ for each $t\ge1$, Algorithm  \ref{QM2_alg} is a  $\delta$-true $\gamma$-threshold estimator under \QMTWO.
	
\end{lemma}

\subsection{Query Complexity Analysis}
The following theorem gives an upper bound on the query complexity of our proposed estimator in Algorithm \ref{QM2_alg}.


\begin{theorem}\label{QM2ub}
Let $\mathcal{A}$ denote the estimator in Algorithm~\ref{QM2_alg} with  $\beta^t = \log(4kt^2 / \delta)$ for each $t\ge 1$ and let $Q_{\delta,\gamma}^{\mathcal{P}}(\mathcal{A})$ be the corresponding query complexity for a given distribution $\mathcal{P}$ under \QMTWO. We define $q$ as $\min\left\{k,{T'=\big(\log (2k/\delta)\big)/{\big(\log\frac{1}{1-\gamma}}\big)}\right\}$. Then,  we have
	\begin{align*}
	 Q_{\delta,\gamma}^{\mathcal{P}}(\mathcal{A})\leq \sum_{i=1}^{m} \max \Biggl\{ \frac{2e\log \Bigl(\sqrt{\frac{4k}{\delta}}\frac{2}{d^{*}(p_i,\gamma)}\Bigr)}{(e-1).d^{*}(p_i,\gamma)} ,{\frac{\log (2k/\delta)}{\log\frac{1}{1-\gamma}} \Biggr\}}
	 \hspace{0 em} + \sum_{i=m+1}^{q} \frac{2e.\log \Bigl(\sqrt{\frac{4k}{\delta}}\frac{2}{d^{*}(p_i,\gamma)}\Bigr)}{(e-1).d^{*}(p_i,\gamma)}, 
	\end{align*}
	with probability at least $1-2\delta$.
\end{theorem}

We essentially argue this as follows. We show that with high probability, each bin created in the course of the algorithm \ref{QM2_alg} corresponding to the support element $i \in \mathcal{S}^{\gamma}_{\mathcal{P}} = \{1,2,...,m\}$ would be out of $\mathcal{C}(t)$ by $\max (T_i,T')$  rounds where $T_i$ is defined as in the argument of Theorem \ref{QM1ub}. Similarly, we show that each bin corresponding to a support element in $\{m+1,...,k\}$ would be out of $\mathcal{C}(t)$ by $T_i$ rounds with high probability. Combining together these two facts gives us the above result, details are provided in the formal proof. 
  


The following result provides a lower bound on the expected number of queries for any $\delta$-true $\gamma$-threshold estimator under the \QMTWO \ model.

\begin{theorem}\label{QM2lb}
	For any $\delta$-true $\gamma$-threshold estimator  $\mathcal{A}$ under \QMTWO, let $Q_{\delta,\gamma}^{\mathcal{P}}(\mathcal{A})$ be the query complexity. Then, we have
	$$\mathbb{E}[Q_{\delta,\gamma}^{\mathcal{P}}(\mathcal{A})]\geq \max_{j\in\{m,m+1\}}\Bigg\{\frac{\log{\frac{1}{2.4\delta}}}{2\times d(p_j||\gamma)}\Bigg\}.$$
\end{theorem}

The proof of the above result involves constructing a $\delta$-true $\gamma$-threshold estimator $\mathcal{A}'$ for the \QMONE \ model using a $\delta$-true $\gamma$-threshold estimator $\mathcal{A}$ under \QMTWO. The lower bound on the query complexity as given in the Theorem \ref{QM2lb} is close to the upper bound in Theorem~\ref{QM2ub} when $\min\{d^{*}(p_m||\gamma),d^{*}(p_{m+1}||\gamma)\} \ll d^{*}(p_i||\gamma)$ $\forall$ $i \notin \{m, m+1\}$ and $T'$ is smaller than $\frac{2e\log \Bigl(\sqrt{\frac{4k}{\delta}}\frac{2}{d^{*}(p_1,\gamma)}\Bigr)}{(e-1).d^{*}(p_1,\gamma)}$. 
In this case, the terms corresponding to $i=m$ and $i=m+1$ dominate in the upper bound on the query complexity of Algorithm \ref{QM2_alg} as given in Theorem \ref{QM2ub}.

We would like to have a lower bound on the query complexity of $\delta$-true $\gamma$-threshold estimators which matches the upper bound more generally and towards this goal, we consider a slightly altered setting which relates closely to the thresholding problem in the Multi Armed Bandit (MAB) setting \cite{MABbook}. \remove{The details are given in Section 1 of the supplementary material.} 
\remove{
	
%


}
\section{QM2 altered setting}
\label{Sec:QM2A}
The altered setting of the \QMTWO \ model we study is motivated from \cite[Appendix C]{shah2019sequentialME}, and we describe it below.

We have $k$ bins, one corresponding to each element of the support set, and a representative element in each bin. Any algorithm proceeds in rounds. In the $t^{th}$ round, the algorithm chooses a subset $\mathcal{C}(t)$ of these bins, and we compare the $t^{th}$ sample with the representative element from each bin in $\mathcal{C}(t)$. Note that these comparisons happen in parallel and the number of queries in round $t$ is the cardinality of the set $\mathcal{C}(t)$. In each round, the oracle response could be $+1$ for some bin $j$ $\in$ $\mathcal{C}(t)$ and $-1$ for all other bins in $\mathcal{C}(t)$ with probability $p_j$, or the response could be $-1$ for all bins in $\mathcal{C}(t)$ with probability $(1 - \sum_{j \in \mathcal{C}(t)} p_j)$. Based on the  oracle responses obtained so far, the algorithm decides whether to stop or to proceed to the next round. When the algorithm decides to stop, it outputs an estimate $\widehat{S}$ of $S_{\mathcal{P}}^{\gamma}$, the set of support elements with a probability above $\gamma$.

Note that this setting is different from our original setting in \QMTWO. Firstly, the number of bins is fixed with one bin corresponding to each element of the support and furthermore, there is a priori one representative element present in each bin. Secondly, in the modified setting, we choose the set $\mathcal{C}(t)$ at the start of each round and all the $|\mathcal{C}(t)|$ replies from the oracle come in parallel. However, in \QMTWO, we we perform queries sequentially in each round and terminate the round as soon as we get a $+1$ response from any one of the bins in $\mathcal{C}(t)$. In spite of these differences, we believe that the query complexity for both these models will be quite similar and as we see below, the alternate setting can be placed in a framework that is fairly well studied and can potentially provide provide pointers towards solving the original problem.  

We look at this new problem as a structured Multi-armed Bandit (MAB) problem \cite{MABbook} where there are $k$ arms, and each arm has a Bernoulli reward distribution with mean $p_i$. From the constraints of our original setup, the means must sum up to $1$ i.e. $\sum_{i} p_i = 1$. In each round we can pull a subset $\mathcal{C}(t)$ of arms and the output is a vector with $-1$ for all arms in $\mathcal{C}(t)$ with probability (1- $\sum_{i \in \mathcal{C}(t)} p_i$) or the output vector has $+ 1$ for some arm $j \in \mathcal{C}(t)$ and $-1$ for all other arms with probability $p_j$. The number of pulls in round $t$ is cardinality of the set $\mathcal{C}(t)$. Based on the responses from the arms, the algorithm decides whether to continue to the next round or stop and output an estimate for the set of arms with mean rewards above $\gamma$. The aim of the algorithm is to correctly identify this set of arms with probability at least (1- $\delta$). The total number of pulls across all rounds is defined as the query complexity of a $\delta$-true $\gamma$-threshold estimator in this setting.

Ideally, we would like to get a tight lower bound on the query complexity for the aforementioned structured MAB problem. The key challenge in doing so is the simplex constraint on the class of mean rewards imposed by $\sum_i p_i =1$. Although we are unable to provide a lower bound for this constraint, we are able to provide a lower bound under a slightly relaxed constraint given by
\begin{equation}{\label{QM2_pi_cond}}
 	p_1 + p_2 +...+p_k + 2 \gamma < 1.
\end{equation}
\begin{theorem}{\label{QM2_alter_lb}}
	For a MAB setting described above where the mean rewards of the individual arms satisfy the condition in equation \ref{QM2_pi_cond}, any $\delta$-true $\gamma$-threshold algorithm has the following lower bound in expectation on the total number of pulls $N$:
	$$
	\mathbb{E}_{\mathcal{P}}[N] \geq \sum_{i=1}^{k} \frac{\log(\frac{1}{2.4\delta})}{2\cdot d(p_i||\gamma)}.
	$$
\end{theorem}         
Note that the expression in the lower bound above is very similar to the upper bound on the query complexity under the \QMTWO \ model derived in Theorem~\ref{QM2ub}. Proving such a lower bound for the structured MAB under the true simplex constraint $\sum_{i} p_i = 1$ is part of our future work. There has been some recent work on similar problems which might provide us some pointers on how to pursue this problem. In particular, say we restrict attention to the class of schemes which compare to a single bin in each round, i.e., $|\mathcal{C}(t)| = 1$ for all $t$. The \textit{thresholding bandit problem} as described above, without the simplex constraint, was studied in  \cite{locatelli2016optimal}, and the optimal query complexity expression turns out to be very similar to the one in Theorem~\ref{QM2_alter_lb}. On the other hand, \cite{simchowitz2017simulator} studies the related problem of identifying the arm with the largest mean reward and derives a tight lower bound on the query complexity under the simplex constraint. 
%
%
\section{Threshold-estimator under \QMTWONOISY}
\label{Sec:QM2N}
Recall that in this model, we make pair-wise comparisons between two samples, and the oracle responses are incorrect with a probability of error $p_e$.
\subsection{Algorithm}
Recall that for the noisy \QMONENOISY \ model, we had established an equivalence to the noiseless \QMONE \ model with a modified threshold $\gamma '$. We would like to establish a similar relation of the \QMTWONOISY \ model to the noiseless \QMTWO \ model studied before. But unlike the \QMONE \ and \QMONENOISY \ models, where the $k$ distinct bins are available apriori, the bins need to be formed using pairwise queries in the \QMTWO \ and \QMTWONOISY \ models. This represents the key challenge under the \QMTWONOISY \ model and in spite of the erroneous pairwise queries, we need to find a reliable method to create bins such that with high probability, all indices in a bin correspond to the same element in the support and different bins correspond to different elements. 

Broadly speaking, our scheme operates in two phases. The objective of the first phase is to extract a collection of bins, each with at least a certain number of indices in it, which includes one corresponding to each element in $\mathcal{S}^{\gamma}_{\mathcal{P}}$. For this phase, we borrow ideas from \cite{mazumdar2017clustering, 8732224} which study the problem of clustering using noisy pairwise queries. The second phase is similar in spirit to the second phase of the estimator for the \QMTWO \ model as described in Algorithm~\ref{QM2_alg}. We create and maintain confidence intervals for each bin by comparing a new sample in each round with representative indices from a subset of the bins extracted in the first phase. We compare the confidence intervals thus created with a a modified threshold $\gamma '$ to reliably identify bins corresponding to elements in $\mathcal{S}^{\gamma}_{\mathcal{P}}$.



We now describe the two phases of the estimator in some more detail, see Algorithm~\ref{QM2n_alg} for the pseudocode. In the first phase, we consider a natural number $T_0$ (defined as in \eqref{T0_defn}) and form a complete graph $\mathcal{G}$ using the first $T_0$ samples, such that each sample corresponds to a vertex of the graph. We query the oracle for each pair of vertices $i$ and $j$ and assign the oracle response as the weight to the edge between vertices $i$ and $j$. For any subgraph $\mathcal{S}$ of the graph $\mathcal{G}$, let $wt(\mathcal{S})$ denote its weight given by the sum of the weights on all the edges in $\mathcal{S}$. The maximum weighted subgraph (MWS) denotes the subgraph $\mathcal{S}$ corresponding to the largest weight $wt(\mathcal{S})$. Starting with the graph $\mathcal{G}$, we repeatedly extract and remove the MWS, as long as the size of the extracted MWS is greater than $S_0 = \frac{\gamma T_0}{4}$. Corresponding to each such extracted MWS, we create a different bin. Let $\mathcal{C}'(T_0)$ denotes the set of bins formed. We will argue that with high probability, each extracted bin corresponds to a unique element in the support and that there is a bin corresponding to each element in $\mathcal{S}^{\gamma}_{\mathcal{P}}$.


The second phase is from round $T_0$ onwards. In this phase, we first choose an index from each bin as a representative element of the bin. In particular, say for each $j\in \mathcal{C}'(T_0)$ we choose index $j_l$ from Bin $j$ as its representative element. For each round $t > T_0$, we consider the $t^{th}$ sample and define $\mathcal{Z}_j^t$ as follows.

\begin{equation}\label{indicator_noisy}
    \mathcal{Z}_j^t=
    \begin{cases}
     1, & \text{if}\ \mathcal{O}(j_l,t)= +1 \\
     0, & \text{otherwise}.
    \end{cases}
\end{equation}

We define  $\tilde{\rho}_j^t$ as the fraction of samples with index larger than $T_0$ that belong to Bin $j$, which can be formally written as follows.

\begin{equation}\label{eq_noisy}
    \tilde{\rho}_j^t = \frac{\sum_{r=T_0+1}^{t} \mathcal{Z}_j^r }{t-T_0}.
\end{equation}

$\mathcal{L}_j(t)$ and $\mathcal{U}_j{(t)}$ denote the lower and upper confidence bounds of Bin $j$  respectively, and are defined as
\begin{equation}
\mathcal{L}_j(t)= {\min}\{ q \in [0,\tilde{\rho}_j^t]: (t-T_0)\times d(\tilde{\rho}_j^t||q) \leq \beta^t  \},\label{eqn_noisylb}
\end{equation}
\begin{equation}
\mathcal{U}_j(t) = {\max}\{ q \in [\tilde{\rho}_j^t,1]: (t-T_0)\times d(\tilde{\rho}_j^t||q) \leq \beta^t  \}.\label{eqn_noisyub}
\end{equation}
In this phase, $\mathcal{C}'(t)$ is defined as 
\begin{align}\label{eqn:c2tnoise}
\mathcal{C}'(t)=\{x| \mathcal{L}_x(t)< \gamma'<\mathcal{U}_x(t) , x \in \mathcal{C}'(t-1) \}.
\end{align}
where, $\gamma'=(1-2p_e)\gamma+p_e$.

In each round $t> T_0$, we go over  the bins in $\mathcal{C}'(t-1)$ one by one,    and query about the samples ${j_l}$ and $t$, $\forall j\in\mathcal{C}'(t-1)$.
We add Sample $t$ to all the  bins for which the oracle provides a  positive response.  In this phase, we also initialize an empty set $\mathcal{S}$, and we update it in each round by adding any bin index $x$ such that $\mathcal{L}_x(t)>\gamma'$. These are the bin indices which the algorithm believes corresponds to elements in $\mathcal{S}_{\mathcal{P}}^{\gamma}$. We run this phase as long as $ \mathcal{C}'(t)$ is non-empty and return the set of bin indices $\mathcal{S}$ upon termination.

The choice of the modified threshold $\gamma '$ above follows from the following observation. Consider a bin $j$ in the second phase and say it corresponds to some support element $i$. Then the expected fraction of indices added to each bin $j$ in the second phase by some round $t > T_0$, denoted by $\tilde{\rho}_j^t$, is equal to $(1-p_e)\times p_i + p_e \times (1 - p_i) = (1-2p_e)p_i + p_e$. 

\begin{algorithm}{\label{QM2n_alg}}
\footnotesize{	
 \SetAlgoLined
 Define $T_0$ as per \eqref{T0_defn} and $S_0= \frac{\gamma.T_0}{4}$\\
 $t$= 1 \\
 {Create a graph $\mathcal{G}$ with just one node labeled $t$.}\\

 \While{$t$ {<} $T_0$}{
 {$t=t+1$\\}
 Create a new node labelled $t$.\\
 
 \remove{. COMMENT: Owing to red line above, won't this create a node for $t = 1$ again Reply: Changed.. }
 $j = 1$ \\
 \While{$j < t$}{

 Create an edge of weight $\mathcal{O}(j,t)$ between nodes $j$ and $t$.\\
 
 $j = j + 1$\\
 }   
  \remove{$t = t+1$}
}
 Extract Maximum Weighted Sub-graph (say $G'$) from this graph.

 \While{$|G'|> S_0$}{
 
 Put all the nodes corresponding to $G'$ in a new bin.
 
 Remove all nodes from $\mathcal{G}$ which were part of $G'$ and edges which were incident to these nodes.\\  
 
 Extract a new Maximum Weighted Sub-graph $G'$ from $\mathcal{G}$.\\  
 
 }

 Denote this extracted set of bins by $\mathcal{C}'(T_0)$
 
 \ForAll{ $j$ $\in$ $\mathcal{C}'(T_0)$}{
 
   Pick representative index $j_l$ $\in$  Bin $j$;
 
 }

 Initialise $\mathcal{S}$ to $\Phi$.\\
 \While{$\mathcal{C}'(t)$ $\neq$ $\phi$}{
 	
 {$t = t+1$\\}	
 \remove{\color{red}Update $\mathcal{C}'(t)$ according to \eqref{eqn:c2tnoise}.}
 
 \ForAll{ $j$ $\in$ $\mathcal{C}'(t-1)$}{
 
 \If{$\mathcal{O}(j_l,t) == +1$}{
   Put index $t$ in bin $i$.
 }
 
 }
 
 Update the empirical estimate $\tilde{\rho}_j^t$ according to \eqref{eq_noisy} and the confidence bounds $\mathcal{L}_j(t)$ and $\mathcal{U}_j(t)$ according to \eqref{eqn_noisylb} and \eqref{eqn_noisyub} respectively $\forall$ bins $\in$ $\mathcal{C}'(t-1)$.
 
 Update $\mathcal{C}'(t)$ according to \eqref{eqn:c2tnoise}.\\
 
 Update $\mathcal{S}$ by adding the bins with $\mathcal{L}_x(t)>\gamma$.
 


 \remove{$t = t + 1$}
 
 }

 Output $\mathcal{S}$.\\
 \remove{\color{red}This pseudocode needs work. Firstly $b(t)$ should be replaced appropriately by $C'(T_0)$. Also 'forall the' can be replaced by `forall'. Also, I think the various while loops are off by one unit because the $t = t+1$ line comes at end. For eg, the while loop for second phase starts at $t = T_0$ again even though that should have been dealt with in first phase. Also first phase seems to have $T_0 - 1$ samples the way it is written now. Also, text said a list $\mathcal{S}$ is maintained which is updated with any bin with $LCB > \gamma '$ but that is not mentioned here at all. We need to be consistent with the text. This needs a careful read and adjustment.}

\caption{Estimator for \QMTWONOISY}
}
\end{algorithm}

The following lemma claims that Algorithm \ref{QM2n_alg} returns the desired set $S_{\mathcal{P}}^{\gamma}$ for any underlying distribution $\mathcal{P}$ with probability at least $1- \delta$.
\begin{lemma}\label{correctness_noisy}
	Given the choice of $\beta^t = \log (\frac{4k(t-T_0)^2}{\delta})$ for each $t > T_0$, where $T_0$ is as defined in \eqref{T0_defn}, Algorithm \ref{QM2n_alg} is a $\delta$-true $\gamma$-threshold estimator under \QMTWONOISY.
\end{lemma}

\subsection{Query complexity Analysis}

%
%
%


The following theorem  provides an upper bound on the query complexity of our proposed estimator in Algorithm \ref{QM2n_alg}. 
\begin{theorem}\label{QM2_noisyub}
Let $\mathcal{A}$ denote the estimator in Algorithm~\ref{QM2n_alg} with  $\beta^t = \log (\frac{4k(t-T_0)^2}{\delta})$ for each $t > T_0$ where $T_0$ is as defined in \eqref{T0_defn}. Let $Q_{\delta,\gamma}^{\mathcal{P}}(\mathcal{A})$ be the corresponding query complexity for a given distribution $\mathcal{P}$ under \QMTWONOISY \ and define $q = \min\{T_0,k\}$, ${p_i}'= (1-2p_e)\times p_i+p_e$ and $\gamma' = (1-2p_e)\gamma+p_e$. Then we have $$Q_{\delta,\gamma}^{\mathcal{P}}(\mathcal{A}) \leq \Bigl(\sum_{i=1}^{q} \frac{2e.\log \Bigl(\sqrt{\frac{4k}{\delta}}\frac{2}{d^{*}({{p}_i}',\gamma')}\Bigr)}{(e-1).d^{*}({p_i}',\gamma')}  + \frac{T_0(T_0-1)}{2} \Bigr)$$

with probability at least $(1- 2\delta)$.
\end{theorem}
\section{NUMERICAL RESULTS}
\label{Sec:Sim}
In this section, we simulate Algorithms \ref{QM1_alg} and \ref{QM2_alg} for the \QMONE \ and \QMTWO \ models respectively, under two different probability distributions. 
\begin{figure}[h]
	\centering
	\begin{minipage}{.5\textwidth}
		\centering
		\includegraphics[width=0.8\linewidth, height=0.2\textheight]{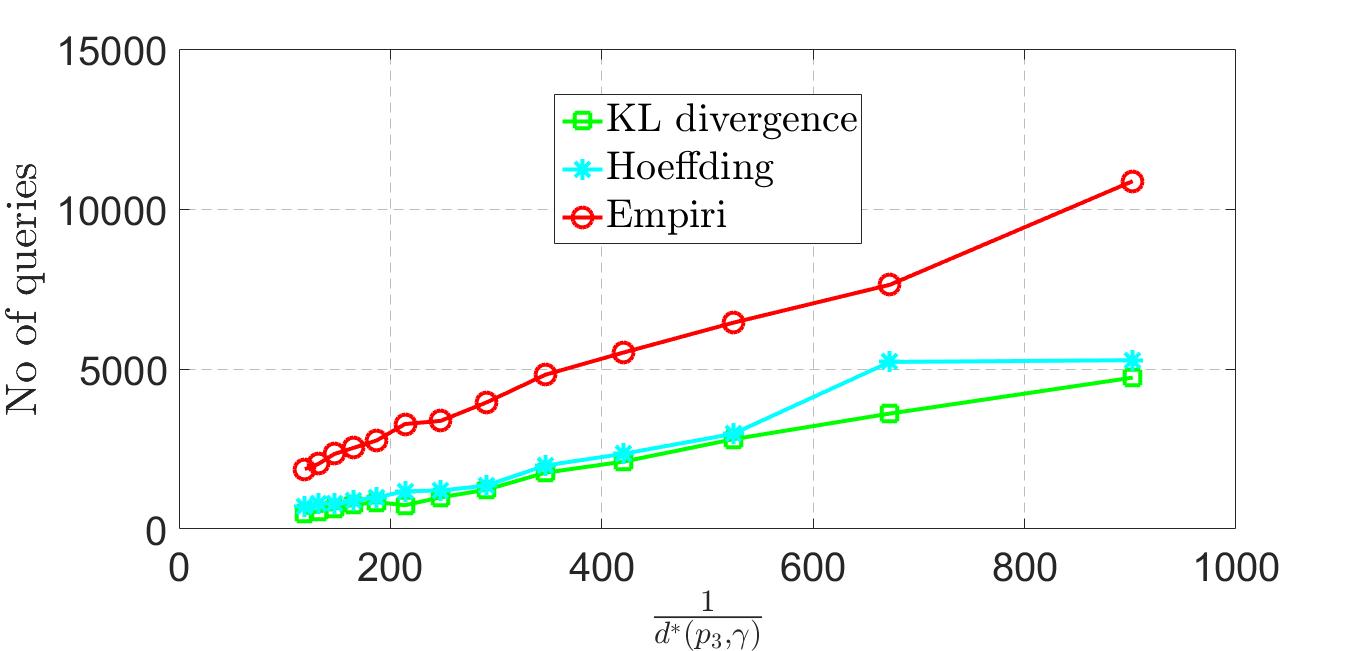}
		\caption{ Query complexity plot of Algorithm \ref{QM1_alg} under different confidence bounds}
		\label{QM1_with_p3}
	\end{minipage}%
        \centering
	\begin{minipage}{.5\textwidth}
		\centering
		\includegraphics[width=0.8\linewidth, height=0.2\textheight]{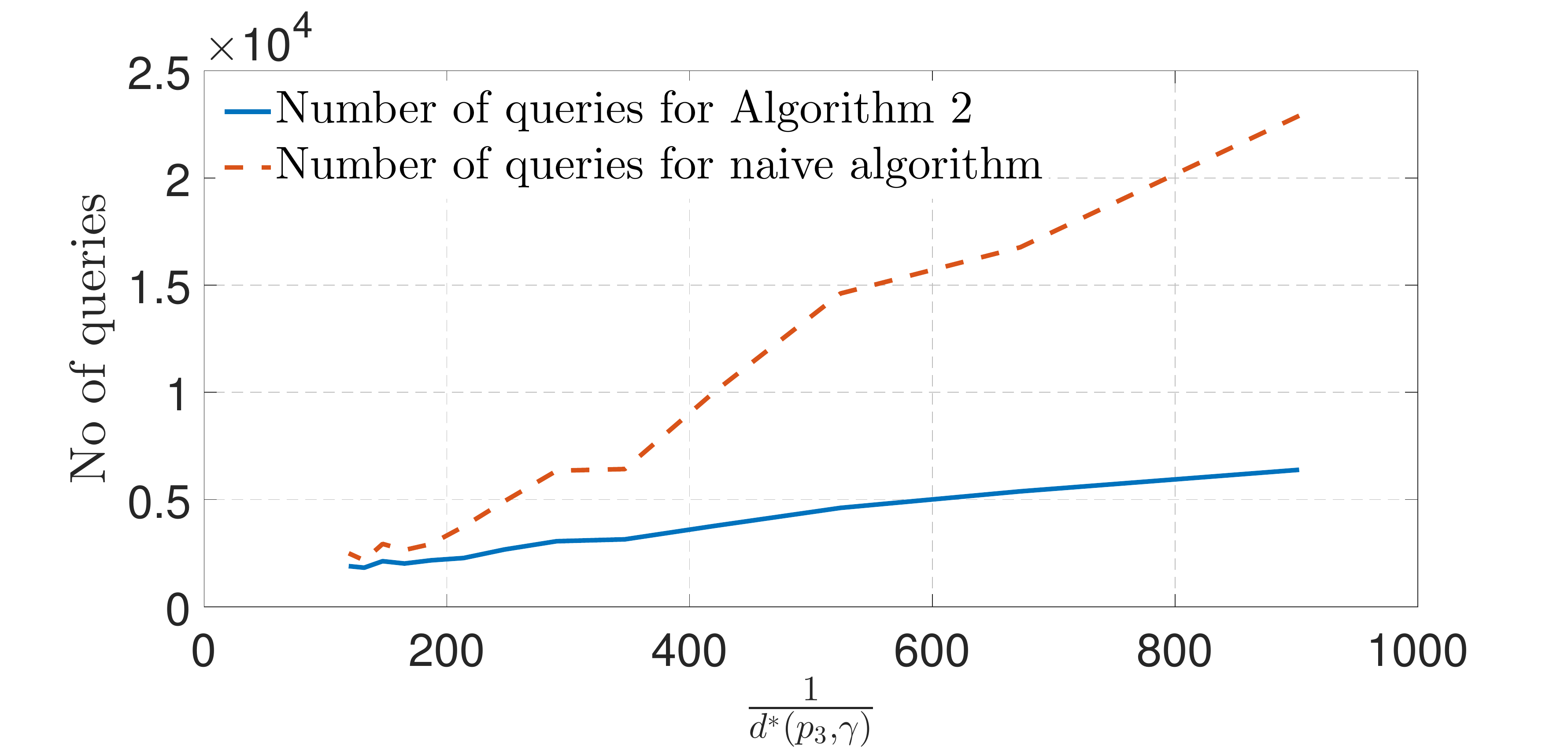}
		\caption{ Query complexity plot for Algorithm \ref{QM2_alg} and a naive algorithm. }
		\label{QM2_with_p3}
	\end{minipage}%
\end{figure}
\begin{enumerate}[label=(\alph*)]
	\item In the first setting, we choose the support size $k=30$ and set $p_1= 0.35$, $p_2=0.28$, vary $p_3$ from $0.13$ to $0.19$, and for all $i=3,4,...,k$, set $p_i= \frac{1-p_1-p_2-p_3}{k-3}$. We set the threshold to $\gamma = 0.1$ and the required error probability $\delta = 0.1$. For each datapoint, we simulate Algorithms \ref{QM1_alg} and \ref{QM2_alg} under \QMONE \ and \QMTWO \ respectively $15$ times each and plot the average number of queries required against $1/{d^{*}(p_3,\gamma)}$ in Fig. \ref{QM1_with_p3} and Fig \ref{QM2_with_p3} respectively. {In Fig. \ref{QM1_with_p3}, we compare the query complexity for Algorithm \ref{QM1_alg} (using KL-divergence based bounds) with those using other popular confidence bounds namely Hoeffding and Empirical Bernstein (used in \cite{shah2019sequentialME})}. As predicted by our theoretical result, the query complexity of Algorithm \ref{QM1_alg} under \QMONE \ increases (almost) linearly with ${1}/{d^{*}(p_3,\gamma)}$ in Fig. \ref{QM1_with_p3}. In Fig \ref{QM2_with_p3}, we compare the query complexity for Algorithm \ref{QM2_alg} with that of a naive algorithm which in each round, queries the next sample with all the bins created so far. {We observe that Algorithm \ref{QM1_alg} the one with KL-divergence based bounds performs better than its counterparts with other popular confidence bounds.}
 	
	We can see that the query complexity of our proposed estimator is much lower since it discards bins as we go along, thus reducing the number of queries. 
	\begin{figure}[h]
		\centering
		\begin{minipage}{0.495\textwidth}
			\centering
			\includegraphics[width=0.8\linewidth, height=0.2\textheight]{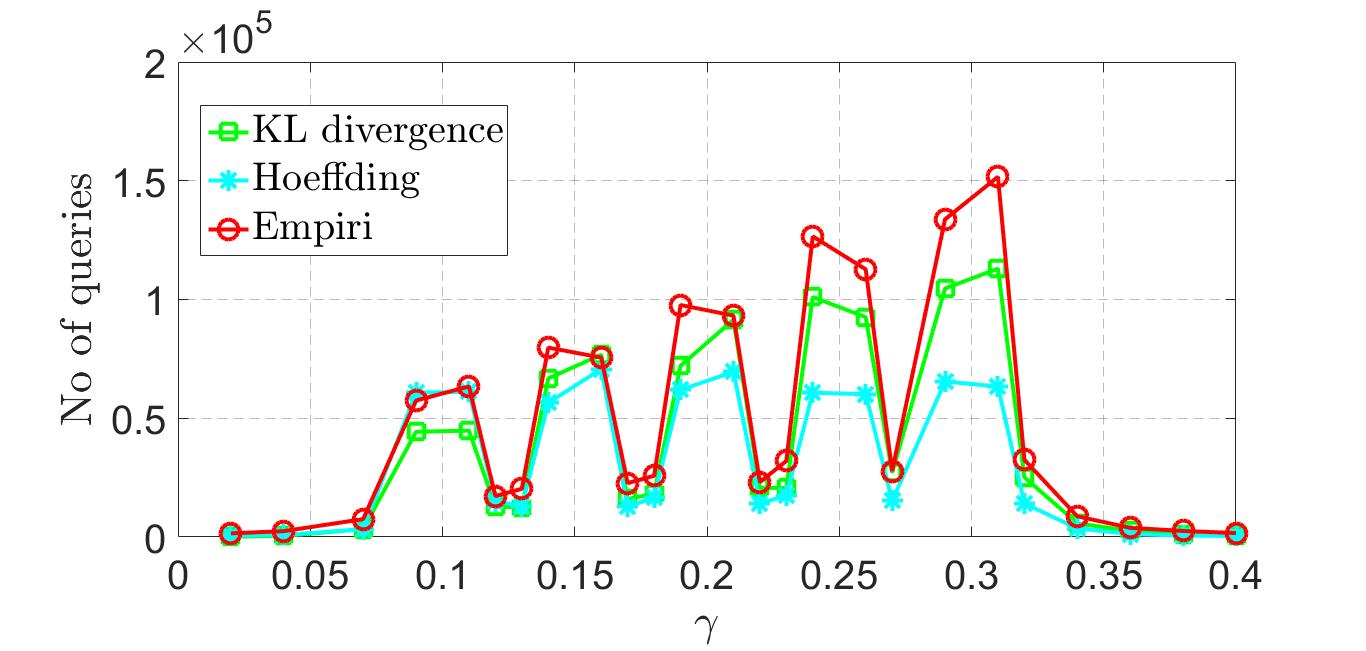}
			\caption{Query complexity plot of Algorithm \ref{QM1_alg} under different confidence bounds with $\gamma$}
			\label{QM1_with_gamma}
		\end{minipage}
	\centering
	\begin{minipage}{0.495\textwidth}
		\centering
		\includegraphics[width=0.8\linewidth, height=0.2\textheight]{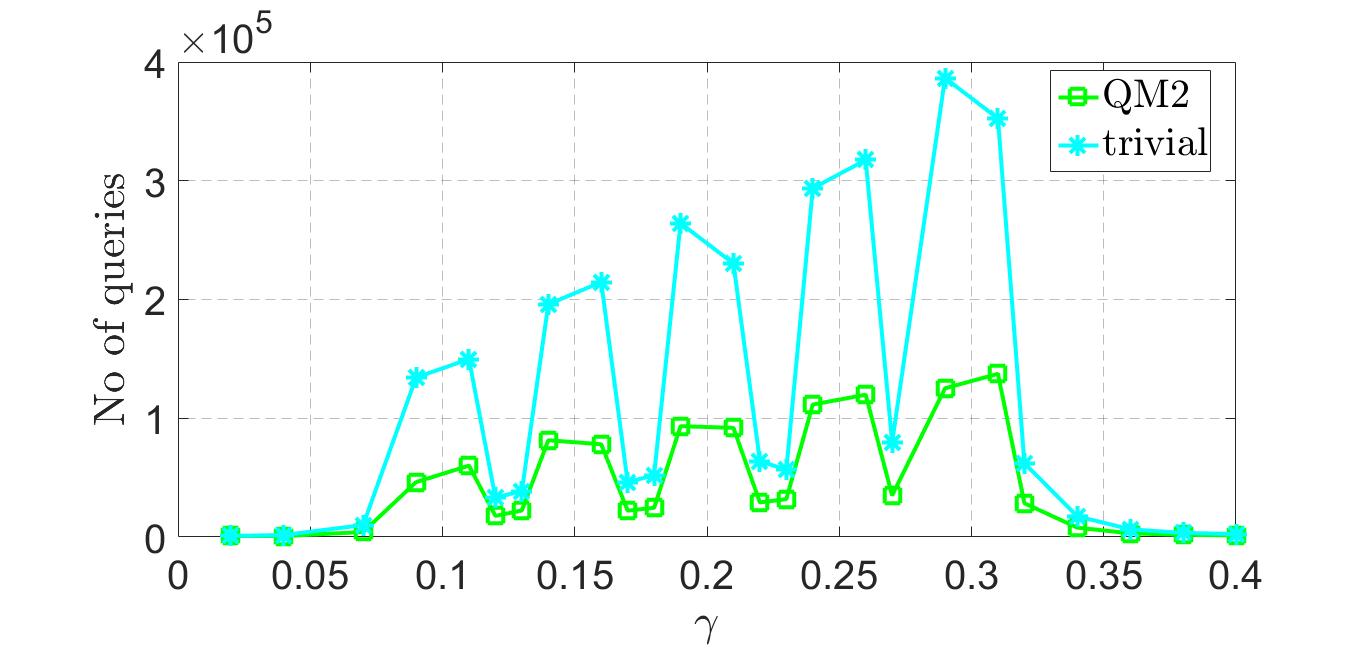}
		\caption{Query complexity plot of Algorithm \ref{QM2_alg} and a naive algorithm with $\gamma$}
		\label{QM2_with_gamma}
	\end{minipage}
\end{figure}
	\item In the second setting, we choose a probability distribution $\{0.3, 0.25, 0.2, 0.15, 0.1\}$ and vary $\gamma$ from $0.02$ to $0.4$. As before, we simulate Algorithms \ref{QM1_alg} and \ref{QM2_alg} under \QMONE \ and \QMTWO \ respectively 15 times each, and plot the average number of queries required against $\gamma$ in Figures~\ref{QM1_with_gamma} \ and \ref{QM2_with_gamma}. {In Fig. \ref{QM1_with_gamma}, we compare the query complexity for algorithm \ref{QM1_alg} (using KL-divergence based bounds) with those other popular confidence bounds namely Hoeffding and Empirical Bernstein (used in \cite{shah2019sequentialME}).} In Fig \ref{QM2_with_gamma}, we compare the query complexity for Algorithm \ref{QM2_alg} with that of a naive algorithm which in each round, queries the next sample with all the bins created so far. The query complexity has multiple peaks, each corresponding to the case where $\gamma$ approaches some $p_i$.

	\begin{figure}[h]
		\centering
		\begin{minipage}{0.495\textwidth}
			\centering
			\includegraphics[width=0.8\linewidth, height=0.2\textheight]{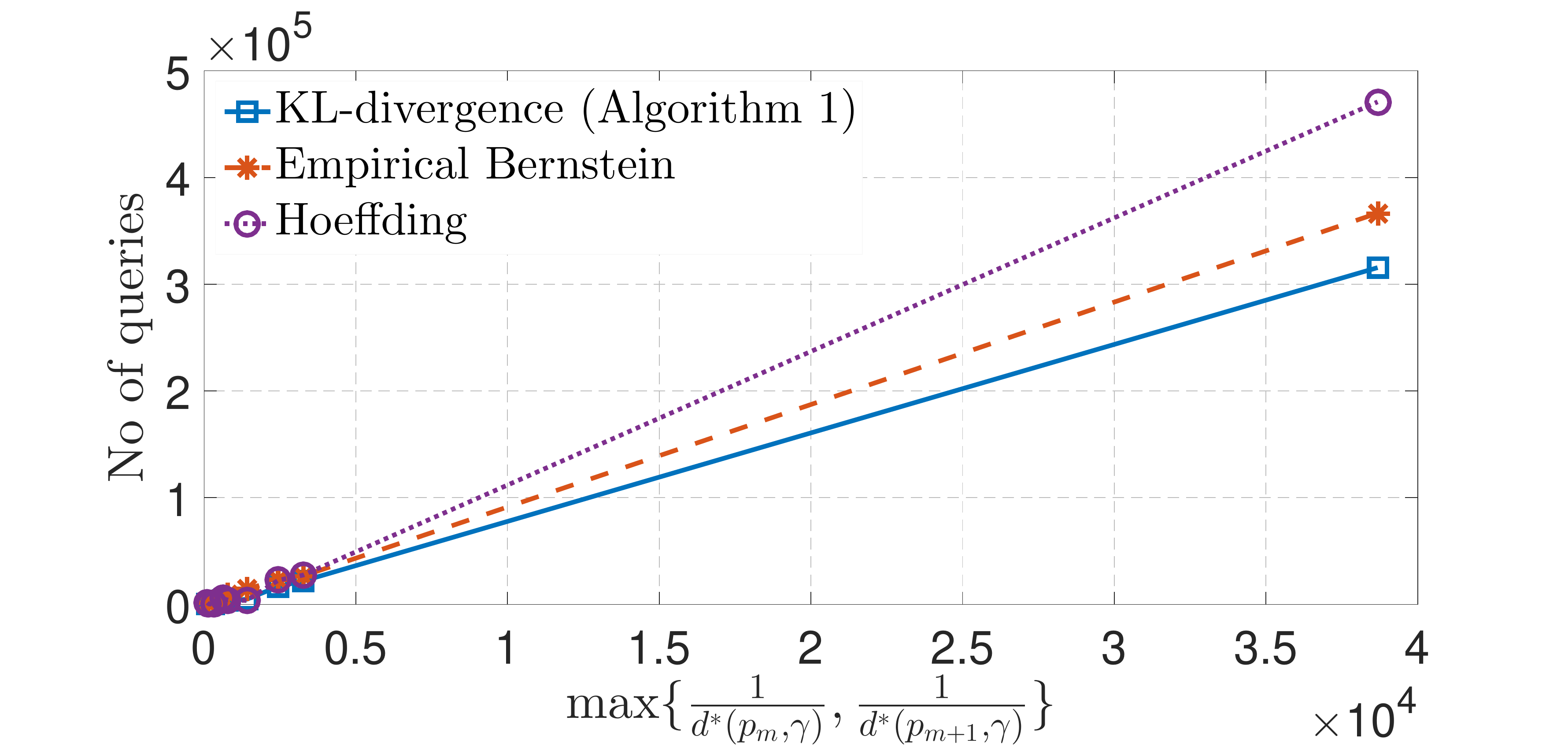}
			\caption{Query complexity plot of Algorithm \ref{QM1_alg} using different confidence bounds for different values of Zipf parameter}
			\label{QM1_with_zipf}
		\end{minipage}
		\centering
		\begin{minipage}{0.495\textwidth}
			\centering
			\includegraphics[width=0.8\linewidth, height=0.2\textheight]{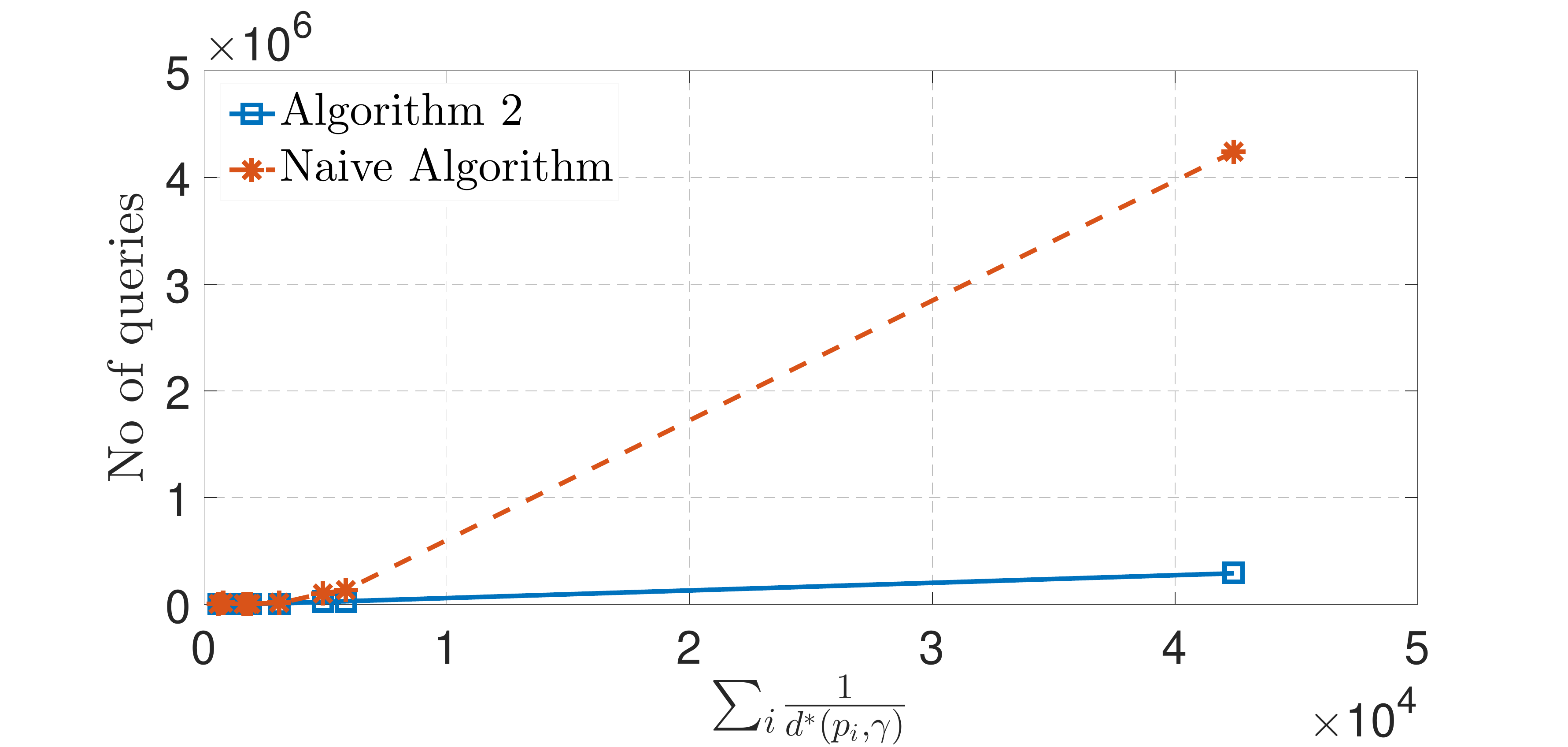}
			\caption{Query complexity plot of Algorithm \ref{QM2_alg} and its naive variant for different values of Zipf paramter}
			\label{QM2_with_zipf}
		\end{minipage}
	\end{figure}
	
\item {In this setting we simulate our algorithms for a fixed value of threshold $\gamma$ = 0.1 against Zipf distributions for various values of Zipf parameter $\beta$ from 0.5 to 4.5. As before we simulate each algorithm 15 times and plot the average query complexity of algorithms corresponding to \QMONE \ and \QMTWO\ with $\max\Bigl\{\frac{1}{d^{*}(p_m,\gamma)},\frac{1}{d^{*}(p_{m+1},\gamma)}\Bigr\}$ and $\sum_i \frac{1}{d^{*}(p_i,\gamma)}$ respectively. In Fig. \ref{QM1_with_zipf}, we compare the query complexity for algorithm \ref{QM1_alg} (using KL-divergence based bounds) with those other popular confidence bounds namely Hoeffding and Empirical Bernstein (used in \cite{shah2019sequentialME}). In Fig \ref{QM2_with_zipf}, we compare the query complexity for Algorithm \ref{QM2_alg} with that of a naive algorithm which in each round, queries the next sample with all the bins created so far. Note that a very similar linear variation is observed similar to those in Fig. \ref{QM1_with_p3} and Fig. \ref{QM2_with_p3}.}

\end{enumerate}
{
\textbf{Comparison on a real-world dataset}: We conduct a clustering experiment on a real-world purchase dataset \cite{clusterdatasets}, where we wish to only identify the clusters with size larger than a given threshold. We benchmark our proposed Algorithm \ref{QM2_alg} for pairwise queries and a naive variant of it with no UCB-based bin elimination against the full clustering algorithm of \cite{mazumdar2017clustering}. We use the dataset in \cite{clusterdatasets} to create a set of nodes (denoting the products in this case) with a label attached to each node such that all nodes attached with a common label represents a set of items belonging to the same product category. 

From the given dataset, we chose the top $k = 100$ clusters for our experiment so that the total number of items is $n= 9,07,101$. The size of the largest cluster is $53,551$ and we chose our threshold size as $17,688$ with $11$ clusters having size larger than it, with the sizes for the $11$-th and $12$-largest clusters being $19390$ and $16298$ respectively. For 99\% target confidence, Algorithm \ref{QM2_alg} terminated with $\sim$ $2.85 \times 10^6$ pairwise queries. In contrast, even for a target confidence of  80\%, the naive variant of Algorithm \ref{QM2_alg}'s where all bins are queried in every round required $\sim$ $17 \times 10^6$  queries ($6\times$ more). On the other hand, the algorithm that does the full clustering first is expected to take around $nk \sim 90 \times 10^6$ queries ($31\times$ more).

{ We also benchmark the performance of our schemes on the Movielens dataset\\ (https://grouplens.org/datasets/movielens/), with each movie associated with its most popular tag, and let each tag represents a cluster. We consider the top 100 clusters which contain 15,241 movies and choose our threshold as 409 with exactly 3 clusters above it. For the QM1 model and with 99 \% confidence, Algorithm 1 required 2,05,394 queries whereas its variants with Hoeffding and Empirical Bernstein - based confidence intervals needed 2,37,346 and 3,06,976 queries respectively. Under the QM2 model, Algorithm 2 terminated after 8,22,124 queries whereas its naive variant required 61,26,357 queries even under 80\% confidence. Also, the full clustering scheme of Mazumdar and Saha 2017 is expected to take around $nk$=15,24,100 queries}.
 
}

\section{Proofs}
\label{Sec:Proofs}
\subsection{Proof of Lemma \ref{Correctness_QM1}}
We use the following lemmas to prove Lemma \ref{Correctness_QM1}. The following lemma is motivated from Lemma 4 in \cite{pmlr-v30-Kaufmann13}.  
\begin{lemma}\label{lemma_cb}
	 Let $l_i(t)$ and $u_i(t)$ be the lower and upper confidence bounds respectively of bin index $i$ and are defined in equations \eqref{eqnlb} and \eqref{eqnub} respectively. Then,
	 \begin{center}
	 	$\mathbb{P}(l_i(t) > p_i) \leq \exp(-\beta^t)$,\\
	 	$\mathbb{P}(u_i(t) < p_i) \leq \exp(-\beta^t)$.
	 \end{center}  
\end{lemma}
\begin{proof}
If $l_i(t)=0$, then the corresponding bound is trivial. For $l_i(t) > 0$, we prove $\mathbb{P}(l_i(t) > p_i)\leq \exp(-\beta^t)$. From equation \eqref{eqnlb}, the event $l_i(t)>p_i$ implies that $t\times d(\tilde{p}_i^t||p_i) > \beta^t$. From the properties of continuity and monotonicity of KL-divergence, there exists $x$ such that $p_i < x < \tilde{p}_i^t$ and $t\times d(x||p_i)=\beta^t$. Thus, we have $\mathbb{P}(l_i(t) > p_i) \le $ $\mathbb{P} (\tilde{p}_i^t >x)\overset{(a)}{\leq}\exp(-t\times d(x||p_i))=\exp(-\beta^t)$, where $(a)$ follows from the Chernoff bound for Binomial random variables, see \cite[Section 1.3]{Concentration_Inequalities}.
	By following similar arguments, we can also  prove that $\mathbb{P} (u_i(t) < p_i)\leq \exp(-\beta^t)$.
\end{proof}
Now we restate and prove  Lemma \ref{Correctness_QM1}.
\begin{lemma*}
	Given the choice of $\beta^t = \log(2kt^2 / \delta)$ for each $t\ge 1$, Algorithm \ref{QM1_alg} is a $\delta$-true $\gamma$-threshold estimator. 
\end{lemma*}
\begin{proof}
 \remove{
We first argue the termination of the algorithm. From Lemma \ref{termination}, we know there exists $T_m$ such that  $t.d(\tilde{p}_i^t||\gamma) > \beta^t$ $\forall$ $t>T_m$ $\forall$ $i$. This would imply from \eqref{eqnlb} and \eqref{eqnub} that either both $u_i(t)$ and $l_i(t)$ is greater than $\gamma$ or both are less than $\gamma$, thus the algorithm terminates with probability 1.  
 }
\remove{\color{red}Let $\mathcal{E}_i^t$ denote the event that $p_i$ does not lie in [$l_i(t)$, $u_i(t)$] at time $t$.} 
For each $i \leq m$, let $\mathcal{E}_i^t$ denote the event that $p_i > u_i(t)$ at time $t$. On the other hand, for $i > m$, let $\mathcal{E}_i^t$ denote the event that $p_i < l_i(t)$ at time $t$. From Lemma \ref{lemma_cb}, we have that $\mathbb{P} (\mathcal{E}_i^t) \leq \delta / 2kt^2$. \remove{= \delta / 2kt^2.}

Let $\mathcal{E}$ denote the event that there exists a pair ($i$, $t$) such that $p_i$ lies above $u_i(t)$ if $i\leq m$, $t \geq 1$ or $p_i$ lies below $l_i(t)$ if $i>m$, $t \geq 1$.\remove{, $i \in \{1,2,\dots,k\}$.} Then, 
\begin{equation}
     \mathbb{P}[\mathcal{E}] = \displaystyle \mathbb{P}[ \bigcup_{i,t} \mathcal{E}_{i}^{t} ] \leq \sum_{i,t} \mathbb{P} [\mathcal{E}_{i}^{t}] \leq \sum_{i,t} \frac{\delta}{2kt^2} \leq \delta.
\end{equation}

We now argue that if the event $\mathcal{E}^{c}$ holds true, the algorithm will correctly return the desired set of support elements $\mathcal{S}_{\gamma}^{\mathcal{P}}$. The termination condition of the algorithm specifies that for each bin, either the LCB lies above $\gamma$ or the UCB lies below $\gamma$. Since we return the set of all bins which have their LCB above $\gamma$ as the estimate $\widehat{S}$ for $\mathcal{S}_{\gamma}^{\mathcal{P}}$, and the event $\mathcal{E}^{c}$ ensures that for each bin index $i\leq m$, $p_i \leq u_i(t)$, and for each index $i>m$ $p_i \geq l_i(t)$, the correctness of our estimator is guaranteed. In particular, the LCB can be greater than $\gamma$ only for indices $i$ such that $p_i > \gamma$. Similarly, the UCB can be smaller than $\gamma$ only for indices $j$ such that $p_j < \gamma$. 

Thus, $\mathbb{P}_{\mathcal{P}}[\widehat{S} = \mathcal{S}_{\gamma}^{\mathcal{P}}] \geq \mathbb{P}_{\mathcal{P}}[\mathcal{E}^c]$ $\geq$ $1-\delta$ and hence Algorithm \ref{QM1_alg} is a $\delta$-true $\gamma$-threshold estimator.

\remove{\color{red} Comment: Need a line somewhere to argue that the algorithm will terminate. Not very sure whether we can argue that. Have to give more thought...}

\end{proof}

\subsection{Proof of Theorem \ref{QM1ub}}
We use the following lemma to prove Theorem \ref{QM1ub} which effectively characterises the number of rounds by which a bin numbered $i$ is classified, i.e. either its LCB goes above $\gamma$ or UCB goes below $\gamma$. 
\begin{lemma}\label{lemma_exittime}
	Let $T_i$ be the smallest positive integer such that $\forall$ $t  > T_i$ the condition $t \times d^{*} (p_i,\gamma) >\beta^t$ is satisfied. For any $t > T_i$, the bin corresponding to support element $i$ is not classified as either above $\gamma$ or below $\gamma$ after $t$ rounds with probability at most $\exp(-\beta^{t})$.
\end{lemma}
\begin{proof}

Consider the event that a bin numbered $i$ for which $p_i < \gamma$ is not yet classified after $t > T_i$ rounds. If this happens, the upper confidence bound $u_i(t)$ is still above $\gamma$ and the lower confidence bound $l_i(t)$ is still below $\gamma$.  From equation \eqref{eqnlb} and \eqref{eqnub}, the event $l_i (t) < \gamma < u_i (t)$ 
implies that $t \times d(\tilde{p}_i^{t}||\gamma) <\beta^{t}$. 
\remove{
	Hence, 
	\begin{align*}
	\mathbb{P}(l_i (T_i)<\gamma < u_i (T_i)) \  
	\leq \ &  \mathbb{P}(T_i \times d(\tilde{p}_i^{T_i}||\gamma)<\beta^{T_i}) 
	\end{align*}
} 
{This implies that there exists $y$ such that $y < \gamma$ and $y < \tilde{p}_i^{t}$ such that $t \times d(y||\gamma) =\beta^{t}$. On the other hand, from the definition of $T_i$ in the statement of the lemma and given $p_i < \gamma$ and $t>T_i$, we have $y>p_i$. Thus, we have $p_i < y < \min\{\gamma,  \tilde{p}_i^{t}\}$.  
	
	Given the series of implications mentioned above, we have the following series of inequalities: 	
	\begin{align}
	\nonumber
	\mathbb{P}(l_i (t)<\gamma < u_i (t) ) \  \leq & \ \mathbb{P} (t \times d(\tilde{p}_i^{t}||\gamma) <\beta^{t})
	\ \leq  \ \mathbb{P}(\tilde{p_i}^{t} >y ) 
	\ \leq  \ \exp(-t \times d(y||p_i))
	\label{eqn:ineq1}
	\end{align}
	where, the last inequality follows from the Chernoff bound for Binomial random variables, see \cite[Section 1.3]{Concentration_Inequalities}.  
	%
	Since  $t \times d^{*} (p_i,\gamma) >\beta^{t}$  and {$t \times d(y||\gamma) =\beta^{t}$}, we have $t \times d(y||p_i) >\beta^{t}$. Then, \remove{from equation \eqref{eqn:ineq1}}we have 
	$$
	\mathbb{P}(l_i(t) <\gamma < u_i (t) ) \   
	\leq \ \exp(-\beta^{t}),
	$$
	thus proving the statement of the lemma for all bin indices $i$ such that $p_i < \gamma$. Similar arguments can be used to prove the result for the bins $j$ with $p_j > \gamma$. 
}
\end{proof}
\remove{\color{red} Comments: I have some doubts on the proof of Lemma 6 as written here. Had some suggestions which I highlighted in red. I believe $T_i$ should be bounded by the larger root of the equation not smaller root and the function is not decreasing $\forall$ $t$ only decreases for sufficietly large $t$.}
\begin{lemma}\label{lemma_timeupper}
Given $\beta^t = \log(2kt^2/ \delta)$ for $t \ge 1$, let $T_i$ be the smallest positive integer such that $\forall$ $t > T_i$ the condition
		$t \times d^{*} (p_i,\gamma) >\beta^t$ is satisfied. Then, setting $a_i = d^{*}(p_i,\gamma)/2$ and $b =  \log(2k /\delta) / 2$, we have
		$$
		T_i \le \frac{e(b-\log(a_i))}{(e-1)a_i} .
		$$
\end{lemma}
%
\begin{proof}
	%
	%
	%
	%
	%
	%
	Since $\frac{\beta^t}{t} = \frac{\log(2kt^2/ \delta)}{t}$ is decreasing in $t$ for sufficiently large $t$. \remove{\color{red}sufficiently large $t$ Comment NK- Checking manually for small $t$, it seems it is always decreasing..so are there multiple roots for this eqn
		
	Reply - SS: Yeah, there are two roots for this equation in most of the cases. Upto a very small value of $t=t_0$  the function increases very fast-, then rapidly decays. However $t_0$ is less than 2 (worst case for $k=1,\delta=1$), 	
		}Thus, it is clear that $T_i$ is upper bounded by the largest root $r^*$ of the equation $t \times d^{*} (p_i,\gamma) = \log(2kt^2 / \delta)$. Letting $a_i$ = $d^{*}({p}_i,\gamma)/2$, $b$ = $\log(\frac{2k}{\delta}) / 2$, this largest root is given by $r^* = \frac{-1}{a_i}\times W_{-1}(-a_i e^{-b})$ where $W_{-1}(y)$ provides the smallest root of $xe^x=y$ for $y < 0$ and denotes the Lambert function \cite{Lambert}.
	
	From  \cite[Theorem 3.1]{Lambert}, we have $W_{-1}(y)> -\frac{e}{e-1}\log(-y)$ and thus $T_i \le r^* \le \frac{e(b-\log(a_i))}{(e-1)a_i}$. This concludes the proof of the lemma.
	\remove{
		Consider the equation $t \times d^{*} (p_i,\gamma) = \log(2kt^2 / \delta)$ which can be written as $t\times a_i = b+\log t$ and can be solved as $\frac{-1}{a_i}\times W(-a_i e^{-b})$ where  $a_i$ = $d^{*}({p}_i,\gamma)/2$, $b$ = (1/2)* $\log(\frac{2k}{\delta})$, and W denotes Lambert function. Note that the lambert function denotes roots for the equation $xe^x=y$ as a function of $y$. Note that the Lambert function takes 2 values for $y<0$ and let $W_{-1}(y)$ denote its smaller root of equation $xe^x=y$ (only defined for $y<0$).
		
		Using Theorem 3.1 of \cite{Lambert}, $W_{-1}(y)>-\frac{e}{e-1}\log(-y)$ and $xe^x>y$ $\forall x<W_{-1}(y)$. By applying this to  $\frac{-1}{a_i}\times W_{-1}(-a_i e^{-b})$, we can say it is upper bounded by $\frac{e(b-\log(a_i))}{(e-1)a_i}$.
		Note that $T_i$ as defined in Lemma \ref{lemma_timeupper} is $W_{-1}(-a_i e^{-b})$ which has been upper bounded by $\frac{e(b-\log(a_i))}{(e-1)a_i}$.
		
		{\color{red} The above proof is currently in a really convoluted form and has to have a simpler exposition. It should be a simple thing. Please take another pass, even though I have a sense of what the argument is, its been written in such a winded form that I can't follow it}
	}
\end{proof}
Now, we will restate and prove Theorem \ref{QM1ub}.
\begin{theorem*}
	Let $\mathcal{A}$ denote the estimator in Algorithm~\ref{QM1_alg} with  $\beta^t = \log(2kt^2 / \delta)$ for each $t\ge 1$ and let $Q_{\delta,\gamma}^{\mathcal{P}}(\mathcal{A})$ be the corresponding query complexity for a given distribution $\mathcal{P}$ under \QMONE. Then,  we have
	$$Q_{\delta,\gamma}^{\mathcal{P}}(\mathcal{A})\leq \max_{j\in\{m,m+1\}} \Biggl\{\frac{2e\log \Bigl(\sqrt{\frac{2k}{\delta}}\frac{2}{d^{*}(p_j,\gamma)}\Bigr)}{(e-1)d^{*}(p_j,\gamma)} \Biggr\},$$
	with probability at least $1-\delta$.
\end{theorem*}
\begin{proof}
	\remove{
		{\color{red}Comment: Before editing the proof here, I want to understand two things: a) Why is the probability in the below paragraph $\frac{\delta}{2k{T_i}^2}$ and not $\frac{\delta}{4k{T_i}^2}$; and (b) why do we target the non-classification probability by $T_i$ for bin $i$ to be $\frac{\delta}{2k{T_i}^2}$ and not $\frac{\delta}{2k}$. Shouldn't this be sufficient since just a union bound over the various bins is needed and this bound will also suffice.}
	}
	From Lemmas \ref{lemma_exittime} and \ref{lemma_timeupper}, we have that for each $i \in \{1,2,\ldots,k\}$, the probability that bin $i$ has its UCB above $\gamma$ and LCB below $\gamma$ beyond $\frac{e.(b-\log(a_i))}{(e-1).a_i}$ rounds is bounded by $\frac{\delta}{2k{T_i}^2}$. Here, $a_i$ = $d^{*}({p}_i,\gamma)/2$ and $b  =  \frac{1}{2} \cdot \log(\frac{2k}{\delta})$.
	
	Taking the worst case number of rounds and applying the union bound over all the $k$ bins, we get that the probability that all bins have been classified by $\max_{i\in [1,k]} \frac{e.(b-\log(a_i))}{(e-1).a_i}$ rounds is greater than or equal to $1-\delta$. It can be verified that $\frac{e.(b-\log(x))}{(e-1)x}$ is decreasing in $x$. Since $a_m \leq a_j$ $\forall$ $j$ $<$ $m$ and $a_{m+1} \leq a_j$ $\forall$ $j$ $>$ $m+1$, the expression for query complexity can be simplified from $\max_{i\in [1,k]} \frac{e.(b-\log(a_i))}{(e-1).a_i}$ to $\max\Bigl\{\frac{e.(b-\log(a_m))}{a_m.(e-1)},\frac{e.(b-\log(a_{m+1}))}{a_{m+1}.(e-1)}\Bigr\}$. 
	\remove{
	{\color{red}. \\
		
		COMMENT: Doesn't the above statement need that the expression for upper bound on $T_i$ be decreasing in $a_i$? Is that obvious? 
		We can check that $\frac{e.(b-\log(a_i))}{a_i.(e-1)}$ is decreasing with $a_i$ by differentiating as long as it is greater than 0.
	}
	}
\end{proof}
\subsection{Proof of Theorem \ref{QM1lb}}{\label{Sec:Thm2_proof}}
The following lemma follows from the proof of  \cite[Theorem3]{shah2019sequentialME} and provides a recipe for deriving lower bounds on the query complexity of  $\delta$-true $\gamma$-threshold estimators. The proof of this lemma follows along similar lines as that for \cite[Theorem3]{shah2019sequentialME} which is based on standard change of measure arguments \cite{kaufmann2016complexity}, and is skipped here for brevity.
\begin{lemma}{\label{shah_paper}}
For any $\delta$-true $\gamma$-threshold estimator $\mathcal{A}$, let $\tau$ be the stopping time of the algorithm. Then, we have 
%
\begin{equation*}
    \mathbb{E}_\mathcal{P}[\tau]\geq \frac{\log{\frac{1}{2.4\delta}}}{\displaystyle \inf_{\mathcal{P}': S_{\mathcal{P}}^{\gamma} \neq S_{\mathcal{P}'}^{\gamma}}D(\mathcal{P}||\mathcal{P}')}
\end{equation*}
\end{lemma}
Now we restate and prove Theorem \ref{QM1lb}.
\begin{theorem*}
	 For any $\delta$-true $\gamma$-threshold estimator  $\mathcal{A}$ under \QMONE, let $Q_{\delta,\gamma}^{\mathcal{P}}(\mathcal{A})$ be the query complexity. Then, we have
   $$\mathbb{E}[Q_{\delta,\gamma}^{\mathcal{P}}(\mathcal{A})]\geq \max_{j\in\{m,m+1\}}\Bigg\{\frac{\log{\frac{1}{2.4\delta}}}{d(p_j||\gamma)}\Bigg\}.$$	
\end{theorem*}
\begin{proof}
	We show that $\mathbb{E}[Q_{\delta,\gamma}^{\mathcal{P}}(\mathcal{A})]\geq \frac{\log{1 / 2.4\delta}}{d(p_m||\gamma)}$, the other inequality follows similarly. 
	
	Lemma~\ref{shah_paper} above requires a choice of distribution $\mathcal{P}'$ such that $S_{\mathcal{P}}^{\gamma} \neq S_{\mathcal{P}'}^{\gamma}$. For some small $\epsilon > 0$, we choose $\mathcal{P}'$ as follows:
	\begin{align*}
		p'_m=\gamma-\epsilon,\quad \text{ and }
		p'_i=\frac{1-\gamma+\epsilon}{1-p_m}p_i , \ \forall \ i\neq m.
	\end{align*}
	Then, we have 
	 \begin{align*}
	 D(\mathcal{P}||\mathcal{P}') &= p_m \log\left(\frac{p_m}{\gamma - \epsilon}\right) + \sum_{i \neq m} p_i \log\left(\frac{p_i}{\frac{1-\gamma+\epsilon}{1-p_m}p_i}\right) \\
	 &= p_m \log\left(\frac{p_m}{\gamma - \epsilon}\right) + (1-p_m)\log\left(\frac{1 - p_m}{1 - \gamma + \epsilon}\right)\overset{(a)}{\approx} d(p_m||\gamma),
	 \end{align*}
	 where $(a)$ follows since $\epsilon$ can be made arbitrarily small. Then, from Lemma~\ref{shah_paper}, we have 
	$$
	\mathbb{E}[Q_{\delta,\gamma}^{\mathcal{P}}(\mathcal{A})] \geq \frac{\log{\frac{1}{2.4\delta}}}{d(p_m||\gamma)}.
	$$
\remove{	We can similarly prove the other lower bound as well.} 
\end{proof}

\subsection{Proof of Lemma \ref{Correctness_QM2}}
%
We use the following lemmas to complete the proof. The first lemma shows that for each index in $\mathcal{S}^{\gamma}_{\mathcal{P}} = \{1,2,\ldots,m\}$,  at least one bin corresponding to it is created in the first phase with high probability, while the second lemma bounds the probability of misclassification of any index $i$. Finally, we argue that two bins corresponding to the same index cannot be returned as part of the estimate $\widehat{S}$ of $\mathcal{S}^{\gamma}_{\mathcal{P}}$ and use the union bound to upper bound the total probability of error of the proposed algorithm.
\begin{lemma}{\label{lemma1:QM2}}
 The probability of the event that an element from the index set $\mathcal{S}^{\gamma}_{\mathcal{P}} =\{1,2,...,m\}$ does not have any bin corresponding to it after the first phase of Algorithm \ref{QM2_alg} is bounded by $\delta/2$.
\end{lemma}
\begin{proof}
 The first phase runs for the first $T' = \frac{\log (\delta/2k)}{\log(1-\gamma)}$ rounds. Consider an element $i$ $\in$ $\{1,2,3,...,m\}$. The probability that $X_t \neq i$ for any $t\ge1$ is equal to ($1-p_i$) $<$ $(1-\gamma)$. Since the samples $X_l$ are i.i.d. $\forall$ $l \in \{1,2,...,T'\}$, we can say that the probability that $X_l \neq i$ $\forall$ $l \in \{1,2,...,T'\}$ is bounded by $(1-\gamma)^{T'}$ $\leq$ $\frac{\delta}{2k}$.    

Applying union bound over all the $m$ elements in $\mathcal{S}^{\gamma}_{\mathcal{P}}$, the probability that 
some element in $\{1,2,...,m\}$ does not have any bin corresponding to it after $T'$ rounds is bounded by $\frac{m\delta}{2k} \leq \frac{\delta}{2}$. 
\end{proof}
\remove{
{\color{red} Changed the lemma and its proof as your suggestion. I feel it has become a bit simpler. However instead of restricting $p_i$ in $[l_i,u_i]$, restricted $p_i$ below $u_i$ for $i\leq m$ and $p_i$ above $l_i$ for $i > m$. This would reduce a factor of 2 inside $\log$ in beta expression if we do it in \QMONE \ as well.}
}
\begin{lemma}{\label{lemma2:QM2}}
For a support element $i\leq m$, define $\mathcal{E}_i^t$ as the event that ${p}_i$ lies above the UCB $u_i(t)$ (defined in \eqref{eqnlb}). Similarly, for any support element $j > m$, define $\mathcal{E}_j^t$ as the event that ${p}_j$ lies below the LCB $l_i(t)$ (defined in \eqref{eqnub}). Then, a bin corresponding to support element $i$ can be misclassified by the algorithm \remove{at some round }only if the event ${\mathcal{E}}_i^t$ holds true for some $t$.
%
\end{lemma}
\begin{proof}
  Before we begin, we need to distinguish between two related quantities. For a bin $b(i)$ corresponding to support element $i$, $\hat{p}_{b(i)}^t$ as defined in equation \eqref{emp_QM2} denotes the fraction of the samples till round $t$ which are placed in the bin; on the other hand, $\tilde{p}_i^t$ as defined in equation \eqref{eq1} denotes the total fraction of samples till round $t$ corresponding to support element $i$. These two can in general be different since during the course of the algorithm, multiple bins can get created corresponding to same support element. 
  
  Next, let $m_u(i)$ denote the $u^{th}$ bin created corresponding to support element $i$. 
   
   \textbf{Case 1: $i\leq m$ }:
    %
   %
      We prove the contrapositive statement here. The event $({\mathcal{E}}_i^t)^c$ $\forall$ $t$ would imply that ${p}_i$ would be less than $u_i(t)$ $\forall$ $t$ which would imply that the first bin corresponding to support element $i$ could never be classified below $\gamma$. As the first bin corresponding to support element $i$ cannot be classified below $\gamma$, multiple bins cannot be created for support element $i$ as per Algorithm \ref{QM2_alg}. This is because multiple bins can be created for a support element only if all previous bins corresponding to the same support element have been classified below $\gamma$. 
      
      Now the empirical probability of the first bin for support element $i$ $(\hat{p}_{m_1(i)}^t)$ is same as the empirical probability of support element $i$ $(\tilde{p}_i^t)$ implying that $\hat{u}_{m_1(i)}(t) = u_{i}(t)$ . Thus, the event $({\mathcal{E}}_i^t)^c$ $\forall$ $t$ would imply that the first and only bin corresponding to support element $i$ is never misclassified. 

  \textbf{Case 2: $i>m$}: Suppose the $l^{th}$ bin corresponding to support element $i$, $m_l(i)$ is misclassified at round $t$, i.e., the LCB corresponding to the bin ${\hat{l}_{m_l(i)}}(t) > \gamma$. We can say $\hat{p}_{m_l(i)}^t \leq {\tilde{p}}_i^t$ where equality would hold true iff $l = 1$. Now $\hat{p}_{m_l(i)}^t \leq \tilde{p}_i^t$ would imply that ${\hat{l}_{m_l(i)}}(t)$ $\leq$ ${l_i}(t)$. Thus ${\hat{l}_{m_l(i)}}(t) > \gamma$ would imply ${l_i}(t) \geq \gamma > p_i$ implying that event $\mathcal{E}_i^t$ occurs. 
    \end{proof}
Now we restate and prove Lemma $\ref{Correctness_QM2}$.  
\begin{lemma*}
	Given the choice of $\beta^t = \log(4kt^2/\delta)$ for each $t\ge1$, Algorithm  \ref{QM2_alg} is a  $\delta$-true $\gamma$-threshold estimator under \QMTWO.
\end{lemma*}
\begin{proof}
  We first argue that there can be no two bins returned as part of the estimate $\widehat{S}$ of $\mathcal{S}^{\gamma}_{\mathcal{P}}$ can correspond to the same support element. We argue this as follows. Suppose there exists such a pair $B_1$ and $B_2$, then both these bins must have been created in the first phase. This is possible only if one of the bins, say $B_1$ was not in $\mathcal{C}(t)$ in the round $t$ when the other bin $B_2$ was created. This would imply that $B_1$ was classified as being below $\gamma$ in the first phase, which would contradict the fact that both the bins were returned.
  
  Next, let $\mathcal{E}_1$ denote the event that some element in $\mathcal{S}^{\gamma}_{\mathcal{P}} = \{1,2,...,m\}$ is not present in any of the bins. From Lemma \ref{lemma1:QM2}, we have $\mathbb{P}[\mathcal{E}_1] < \delta/2$.

  Next, let $\mathcal{E}_2$ denote the event that there exists a misclassified bin corresponding to some support element $i$. From Lemma \ref{lemma2:QM2}, this would imply that $\mathcal{E}_i^t$ occurs for some $(i,t)$ pair. However for $i\leq m$, $\mathcal{E}_i^t$ implies $u_i(t) > p_i$ whose probability can be bounded by Lemma \ref{lemma_cb} to be at most $\frac{\delta}{4kt^2}$. Similarly, for $i>m$, $\mathcal{E}_i^t$ implies $l_i(t)<p_i$ whose probability can be bounded by Lemma \ref{lemma_cb} to be at most $\frac{\delta}{4kt^2}$. {Taking the union bound over all such pairs $(i,t)$, we obtain $\mathbb{P}[\mathcal{E}_2]$ $\leq$ $\frac{\delta}{2}$}.
  
 \remove{ {\color{red} Comment: how? what is the probability for a given pair $(i,t)$, before you take sum ? Reply: Added in the description.
  	
  } 
  
 }   
We can see  that Algorithm \ref{QM2_alg} returns an incorrect set of bins only if at least one of the two events $\mathcal{E}_1$, $\mathcal{E}_2$ has occurred. Applying the union bound on events $\mathcal{E}_1$ and $\mathcal{E}_2$, we get that the probability of Algorithm \ref{QM2_alg} returning an incorrect set of bins is bounded by $\delta$.  
 \end{proof}
\subsection{Proof of Theorem \ref{QM2ub}}
Recall that we defined $a_i = d^{*}({p}_i,\gamma)/2$ and $b = \frac{1}{2} \cdot \log(\frac{4k}{\delta})$. We use the following lemmas to prove Theorem \ref{QM2ub}.
\begin{lemma}{\label{queries with each box}}
\remove{Given that \remove{event $(\mathcal{E}_i^t)^c$ (defined in Lemma \ref{lemma2:QM2}) occurs $\forall$ $i,t$} the algorithm correctly returns all the bins,} Each of the following statements is true with probability at most $\frac{\delta}{k}$.
\begin{itemize}{\Roman{enumii}}
	\item The total number of queries with all the bins representing support element $i\leq m$ is greater than $\max\left\{\frac{e(b-\log(a_i))}{(e-1)a_i},T'\right\}$ and the algorithm {correctly classifies all the bins\remove{\color{red} with no misclassification at any round}.}\remove{corresponding to support element $i$}
	\item The total number of queries with all the bins representing support element $i>m$ is greater than $\frac{e(b-\log(a_i))}{(e-1)a_i}$ and the algorithm correctly classifies all the bins\remove{\color{red} with no misclassification at any round}.
\end{itemize}
\end{lemma}
\begin{proof}
 \remove{\color{red}Comment: Edited the proofs as per our discussion. }
{Consider any bin which contains indices representing support element $i$.}\\ 
 \remove{We know from Lemma \ref{lemma2:QM2} that $(\mathcal{E}_i^t)^c$ $\forall$ $i,t$ implies that the algorithm correctly returns the desired set of bins.}
 \begin{enumerate}
 	\item $ i \leq m$: Let us denote the bin with samples representing support element $i$ as $m_i$. All bins corresponding to support element $i$ being correctly classified and $i \leq m$ would imply that Bin $m_i$ would be classified above $\gamma$. According to Algorithm \ref{QM2_alg}, no bin classified above $\gamma$ would have another bin corresponding to the same support element created again. Therefore, $m_i$ is the first and last bin created for support element $i$ implying $\hat{p}_{m_i}^t=\tilde{p}_i^{t}$.
 	
 \remove{\color{red}	Then Lemma \ref{lemma_exittime} tells us that after $\frac{e(b-\log(a_i))}{(e-1)a_i}$ rounds, Bin $m_i$ would have LCB below $\gamma$ and UCB above $\gamma$ with probability at most $\delta/ 2k$. However we know that no bin with UCB above $\gamma$ comes out of $\mathcal{C}(t)$ before $T'$ rounds.
 	
 	{Thus with probability at most $\delta/ 2k$, we say the bin $m_i$ would not exit the set $\mathcal{C}(t)$ after \\ $\max\left\{\frac{e(b-\log(a_i))}{(e-1)a_i},T'\right\}$ queries with bins representing support element $i$ and the algorithm correctly classifies $m_i$. \remove{the set of bins corresponding to support element $i$.}}}
 	
 	At round ${T_i}'  = \max\left\{\frac{e(b-\log(a_i))}{(e-1)a_i},T'\right\}$, Algorithm \ref{QM2_alg} would have certainly proceeded to the second phase. Hence, the event that bin $m_i$ does not drop out of the subset of bins $\mathcal{C}(t)$ against which new samples are compared after ${T_i}'$ rounds would imply that the condition in equation \eqref{eqn:c2t} is satisfied after ${T_i}'$ rounds. This requires ${l}_i({T_i}') <\gamma < {u}_i({T_i}')$, which from Lemma \ref{lemma_exittime} and Lemma \ref{lemma_timeupper} is true with probability at most $\frac{\delta}{2k}$.

   \item $i > m$: In the course of a run of Algorithm \ref{QM2_alg}, there might be multiple bins created for support element $i>m$ even if the algorithm correctly classifies all the bins. Let $m_l(i)$ denote the the $l^{th}$ bin created for support element $i$, if created. Recall that $\hat{p}_{m_l(i)}^t$ as defined in equation \eqref{emp_QM2} denotes the fraction of the samples till round $t$ which are placed in the bin $m_l(i)$; on the other hand, $\tilde{p}_i^t$ as defined in equation \eqref{eq1} denotes the total fraction of samples till round $t$ corresponding to support element $i$. It is easy to see that $\hat{p}_{m_l(i)}^t \leq \tilde{p}_i^{t}$ where equality holds iff $l = 1$. The bin $m_1(i)$ not being out of $\mathcal{C}(t)$ after $K_i = \frac{e(b-\log(a_i))}{(e-1)a_i}$ rounds\remove{the bin $m_1(i)$ not being misclassified} would imply the event $\hat{l}_{m_1(i)}(K_i)< \gamma < \hat{u}_{m_1(i)}(K_i)$ which is equivalent to the event $l_{i}(K_i)< \gamma < u_{i}(K_i)$ \remove{and whose probability can be bounded using Lemma \ref{lemma_exittime} and \ref{lemma_timeupper} by $\frac{\delta}{2k}$} or the event $\gamma <\hat{l}_{m_1(i)}(K_i) $ which is equivalent to the event $p_i<\gamma< l_i(K_i)$. \remove{whose probability can be bounded by Lemma \ref{lemma_cb} $\forall t$ by $\sum_i \frac{\delta}{4kt^2} \leq \frac{\delta}{2k}$} The probability of the event $p_i<\gamma< l_i(K_i)$ is bounded using Lemma \ref{lemma_cb} by $\frac{\delta}{2k}$. The probability of the event $l_i(K_i)<\gamma<u_i(K_i)$ is bounded by Lemma \ref{lemma_exittime} and Lemma \ref{lemma_timeupper} by $\frac{\delta}{2k}$. 
   
   Thus the probability of the event bin $m_1(i)$ not being out of $\mathcal{C}(t)$ after $K_i$ rounds is bounded by sum of probabilities of the event $p_i<\gamma< l_i(K_i)$ and the event $l_{i}(K_i)< \gamma < u_{i}(K_i)$ which can be bounded by $\frac{\delta}{2k}+\frac{\delta}{2k} \leq \frac{\delta}{k}$.
   
   
    \remove{{\color{red}Comment: Why should the LCB be below $\gamma$ at $K_i$.\\ 
   	Reply-SS : If the LCB is not below $\gamma$, the bin $m_1(i)$ would be misclassified at round $K_i$ as both LCB and UCB are above $\gamma$ though it may not be out of $\mathcal{C}(t)$ if $K_i <T'$. To avoid confusion, I am adding the part no misclassification at any round. }}
   
   Now consider a bin $m_l(i)$, $l>1$ created at some round $t_l$. Since the bin $m_l(i)$ was created at some round $t_l$ and all bins are correctly classified,  bin $m_1(i)$ must have been classified correctly below $\gamma$ in some previous round $t_1 < t_l$. It is easy to see that $\hat{p}_{m_l(i)}^{t_l} < \hat{p}_{m_1(i)}^{t_1}$ which would imply that $\hat{u}_{m_l(i)}^{t_l} < \hat{u}_{m_1(i)}^{t_1} \overset{(a)}{<} \gamma$ and thus the bin $m_l(i)$ would immediately be out of $\mathcal{C}(t)$ after its creation at round $t_l$ and there would be no queries with bin $m_l(i)$ $\forall$ $l>1$. Note that $(a)$ follows since the bin $m_1(i)$ is correctly classified below $\gamma$.
   
   
    \remove{\color{red}What does classification mean here? That might not happen till the end of phase 1. Also, is it clear that LCB and UCB are monotonic in time 
   	Reply-SS: UCB going below $\gamma$ immediately implies being out of $\mathcal{C}(t)$ irrespective of whichever phase the algorithm is in. Also since we know that the bounds i.e $\frac{\beta^t}{t}$ is decreasing, we can say if $\hat{p}_{m_l(i)}^{t_l} < \hat{p}_{m_1(i)}^{t_1}$ then, $\hat{u}_{m_l(i)}^{t_l} < \hat{u}_{m_1(i)}^{t_1}$  }
   
   Thus the probability of total number of queries with all the bins denoting support element $i>m$ is greater than $\frac{e(b-\log(a_i))}{(e-1)a_i}$ and the algorithm correctly classifies all the bins is upper bounded by $\frac{\delta}{k}$.   
 	
 	
 \end{enumerate} 
%
%
%
%
%
%
%
%
%
%
\end{proof}
Let us restate and prove Theorem \ref{QM2ub}
\begin{theorem*}
	Let $\mathcal{A}$ denote the estimator in Algorithm~\ref{QM2_alg} with  $\beta^t = \log(4kt^2 / \delta)$ for each $t\ge 1$ and let $Q_{\delta,\gamma}^{\mathcal{P}}(\mathcal{A})$ be the corresponding query complexity for a given distribution $\mathcal{P}$ under \QMTWO. We define $q$ as $\min\left\{k,\frac{\log (\delta/2k)}{\log(1-\gamma)}\right\}$. Then,  we have
	\begin{align*}
	 Q_{\delta,\gamma}^{\mathcal{P}}(\mathcal{A})\leq \sum_{i=1}^{m} \max \Biggl\{ \frac{2e\log \Bigl(\sqrt{\frac{4k}{\delta}}\frac{2}{d^{*}(p_i,\gamma)}\Bigr)}{(e-1).d^{*}(p_i,\gamma)} ,\frac{\log (\delta/2k)}{\log(1-\gamma)} \Biggr\}
	 \hspace{0 em} + \sum_{i=m+1}^{q} \frac{2e.\log \Bigl(\sqrt{\frac{4k}{\delta}}\frac{2}{d^{*}(p_i,\gamma)}\Bigr)}{(e-1).d^{*}(p_i,\gamma)}, 
	\end{align*}
	with probability at least $1-2\delta$.
\end{theorem*}
\remove{In the proof, we essentially bound the error probability by first conditioning on the correctness of algorithm and then bound the net error probability of algorithm.} 
\begin{proof}	
The total number of bins created during the course of the algorithm must be bounded by $T' = \frac{\log (\delta/2k)}{\log(1-\gamma)}$ since new bins are created only in the first phase of the algorithm and at most one new  bin can be created in each round. Thus, the number of bins corresponding to distinct support elements is upper bounded by $q = \min\left\{k,\frac{\log (\delta/2k)}{\log(1-\gamma)}\right\}$. Furthermore, if the algorithm returns a correct set of bins, at least one bin corresponding to each element in $\{1,2,...,m\}$ must have been created.

\remove{
 However, the total number of bins corresponding to different support elements must be bounded by $T'$. Thus the total number of bins corresponding to distinct support elements created must be upper bounded by $q=\min(k,T')$. 
}
\remove{Given that the algorithm returns a correct set of bins,} Using Lemma \ref{queries with each box} and taking the union bound over all support elements $i \in \{1,2,...,k\}$, the probability of the event that the algorithm correctly classifies all the bins \remove{with no misclassification at any round}\remove{correctly classifies all the bins}and there is\remove{some support element has a bin representing itself in $\mathcal{C}(t)$} an unclassified bin after $\sum_{i=1}^{m} \max\left\{\frac{e(b-\log(a_i))}{(e-1)a_i},T'\right\}$ + $\sum_{i=m+1}^{q} \frac{e(b-\log(a_i))}{(e-1)a_i}$ queries is at most ${\delta}$. 


{Also, using Lemma $\ref{Correctness_QM2}$, the probability that the algorithm returns an incorrect set of bins is bounded by $\delta$.} 
\remove{\color{red} Also using Lemma \ref{lemma2:QM2}, the probability that the algorithm misclassifies some bin at some round would be bounded by $\sum_{i,t}\mathbb{P}[\mathcal{E}_i^t] \leq \delta/2$ as done in the proof of Lemma \ref{Correctness_QM2}} Combining together these two observations, we have that the probability that the Algorithm $\ref{QM2_alg}$ does not terminate after 
$\sum_{i=1}^{m} \max\left\{\frac{e(b-\log(a_i))}{(e-1)a_i},T'\right\}$ + $\sum_{i=m+1}^{q} \frac{e(b-\log(a_i))}{(e-1)a_i}$ queries is bounded by 2$\delta$. This completes the proof of the result.
\remove{
Thus, the probability of algorithm having some bin in $\mathcal{C}(t)$ after  $\sum_{i=1}^{m} max(\frac{e(b-\log(a_i))}{(e-1)a_i},T')$ + $\sum_{i=m+1}^{q} \frac{e(b-\log(a_i))}{(e-1)a_i}$  queries is bounded by 2$\delta$ implying .
}
\remove{
{\color{red}Read the above to check if its ok after the edits.
	Reply- SS: Added the sentence highlighted in red to incorporate the slight change in lemma 11.
	
	Comment - SS: Reverted the proof back to original.
	}}
\end{proof}

\subsection{Proof of Theorem \ref{QM2lb}}
Now, we restate and prove Theorem \ref{QM2lb}.
\begin{theorem*}
	For any $\delta$-true $\gamma$-threshold estimator  $\mathcal{A}$ under \QMTWO, let $Q_{\delta,\gamma}^{\mathcal{P}}(\mathcal{A})$ be the query complexity. Then, we have
	$$\mathbb{E}[Q_{\delta,\gamma}^{\mathcal{P}}(\mathcal{A})]\geq \max_{j\in\{m,m+1\}}\Bigg\{\frac{\log{\frac{1}{2.4\delta}}}{2\times d(p_j||\gamma)}\Bigg\}.$$
\end{theorem*}
\begin{proof}
Consider any $\delta$-true $\gamma-$threshold estimator $\mathcal{A}_3$ under \QMTWO \ and let us denote the total number of queries by $\tau$ when the underlying distribution is $\mathcal{P}$. Using the above estimator let us construct a $\delta$ -true $\gamma-$threshold estimator for \QMONE. We create such an estimator $\mathcal{A}'$ by simply querying all the indices involved in \QMTWO. Since we know that $\mathcal{A}_3$ is a $\delta$-true $\gamma-$threshold for \QMTWO, we can argue that $\mathcal{A}'$ would also be a $\delta$-true $\gamma-$threshold for \QMONE \ and thus the query complexity of $\mathcal{A}_3$ would be $2.\tau$.

Thus if the expected query complexity $\mathbb{E}[\tau]$ of $\mathcal{A}_3$ is less than $\max_{j\in\{m,m+1\}}\Bigg\{\frac{\log{\frac{1}{2.4\delta}}}{2\times d(p_j||\gamma)}\Bigg\}$ we can construct an estimator $\mathcal{A}'$ for noiseless query model 1 whose expected query complexity is less than $\max_{j\in\{m,m+1\}}\Bigg\{\frac{\log{\frac{1}{2.4\delta}}}{d(p_j||\gamma)}\Bigg\}$ which contradicts the Theorem \ref{QM1lb}.
\end{proof}
\subsection{Proof of Theorem \ref{QM2_alter_lb}}
The following lemma, which we prove in \ref{first_ineq_proof}, will be used to prove this result.
\begin{lemma}{\label{first_ineq}}
	For all $1 \geq p_t, \gamma \geq 0$, the following inequality holds true.	
	$$p_t \log \Bigl(\frac{p_t}{\gamma}\Bigr)+ (2.\gamma) \log\Bigl(\frac{2.\gamma}{p_t + \gamma}\Bigr) \leq 2 d(p_t||\gamma).$$ 
\end{lemma}
Now we restate and prove Theorem \ref{QM2_alter_lb}.
\begin{theorem*}
	For a MAB setting described above where the mean rewards of the individual arms satisfy the condition in equation \ref{QM2_pi_cond}, any $\delta$-true $\gamma$-threshold algorithm has the following lower bound in expectation on the total number of pulls $N$:
	$$
	\mathbb{E}_{\mathcal{P}}[N] \geq \sum_{i=1}^{k} \frac{\log(\frac{1}{2.4\delta})}{2\cdot d(p_i||\gamma)}.
	$$	
\end{theorem*}
The proof of this result follows along similar lines as that of of \cite[Theorem 7]{shah2019sequentialME}.

\begin{proof}
Consider an estimator $\mathcal{A}$ which can correctly identify the arms with mean reward distribution above $\gamma$ with probability at least $(1 - \delta)$. We consider two such distributions with mean reward profiles as follows:
\begin{align*}
\mathcal{P}  =  (p_1,p_2,..,p_l,..,p_k),\quad \text{ and } \quad
\mathcal{P}' = ({p_1}',{p_2}',..,{p_l}',..,{p_k}').
\end{align*}
where as before, we assume for $\mathcal{P}$ that $p_1 \geq p_2 \geq \ldots \geq p_m > \gamma > p_{m+1} \geq \ldots p_k$. Also, let $l \leq m$ with ${p_l}'=\gamma - \epsilon$ for some small $\epsilon > 0$ and ${p_i}' = p_i$ $\forall$ $i \neq t$. 

Note that the sets of arms with mean reward distribution above $\gamma$ for the distributions $\mathcal{P}$ and $\mathcal{P}'$, $\mathcal{S}^{\gamma}_{\mathcal{P}}$ and $\mathcal{S}^{\gamma}_{\mathcal{P}'}$ respectively, are different.

Recall that, we may decide to pull any subset $S$ of the arms in any round and there can be $2^k$ such subsets. For any subset $S$, with probability $p_j$, the output vector in any round can be $+1$ for some arm $j \in S$ and $-1$ for all other arms; and with probability (1 - $\sum_{i} p_i$) the output vector is $-1$ for all arms . Let ($Y_{S_a,s}$) be the output vector observed while pulling the subset $S_a$ for the $s^{th}$ time. Based on the observations till round $t$, we define the likelihood ratio $L_t$ as follows:
\begin{equation*}
L_t = \sum_{a=1}^{2^k} \sum_{s=1}^{N_{S_a}(t)} \log  \Bigl(\frac{f_{S_a}(Y_{S_a,s})}{f'_{S_a}(Y_{S_a,s})}\Bigr).
\end{equation*}
Here $N_{S_a}(t)$ denotes the number of times the subset of arms $S_a$ was pulled till round $t$. With a slight misuse of notation, we let $N_i(t)$ denote the number of times arm $i$ was pulled till round $t$, which sums over all subsets containing $i$. 

We say
\begin{equation*}
\mathbb{E}_{\mathcal{P}}\Bigl[ \log  \Bigl(\frac{f_{S_a}(Y_{S_a,s})}{f'_{S_a}(Y_{S_a,s})}\Bigr) \Bigr]= D(p_{S_a},{p_{S_a}}'). 
\end{equation*}

Applying Wald's stopping lemma to $L_{\sigma}$ where $\sigma$ is the stopping time associated with estimator $\mathcal{A}$ we have,
\begin{align} 
 \mathbb{E}_{\mathcal{P}}[L_{\sigma}] =  \sum_{a=1}^{2^k} \mathbb{E}_{\mathcal{P}}[N_{S_a}(\sigma)] D(p_{S_a},{p_{S_a}}') 
&\overset{(a)}{\leq} \mathbb{E}_{\mathcal{P}}[N_{l}(\sigma)] \max_{S_a: l \in S_a} D(p_{S_a},{p_{S_a}}'). {\label{expectation_likelihood}}
\end{align}
where $(a)$ follows from the definition of $\mathcal{P}$ and $\mathcal{P}'$ as $D(p_{S_a},{p_{S_a}}')$ is zero for those sets which don't contain arm $l$. Next, for any set $S_a$ such that $l \in S_a$ and $\sum_{i \in S_a \backslash \{l\}} p_i = s$, we have 
\begin{align*}
 D(p_{S_a},{p'_{S_a}})  = p_l.\log\Bigl(\frac{p_l}{\gamma-\epsilon}\Bigr)
\hspace{0em} +(1-(p_l+s))\log\Bigl(\frac{1-(p_l+s)}{1-(\gamma-\epsilon+s)}\Bigr).
\end{align*}
We can show that the second term is increasing with $s$ and hence takes its maximum value when $s = \sum_i{p_i}-p_l$. Thus,
\begin{align*}
&\max_{S_a:t \in S_a} D(p_{S_a},{p_{S_a}}') = p_l \log \Bigl(\frac{p_l}{\gamma-\epsilon}\Bigr) +
 \hspace{0 em}  (1-\sum_i{p_i})\log\Bigl(\frac{1-\sum_i{p_i}}{1-(\gamma-\epsilon+\sum_i p_i - p_l)}\Bigr)\\ & \overset{(a)}{\leq} p_l \log \Bigl(\frac{p_l}{\gamma-\epsilon}\Bigr)+ (2.\gamma) \log\Bigl(\frac{2.\gamma}{(p_l+\gamma+\epsilon)}\Bigr) \overset{(b)}{\approx} p_l \log \Bigl(\frac{p_l}{\gamma}\Bigr)+ (2.\gamma) \log\Bigl(\frac{2.\gamma}{(p_l+\gamma)}\Bigr) \overset{(c)}{\leq} 2 d(p_l||\gamma),
\end{align*}
where $(a)$ follows from the fact that $(1-\sum{p_i}) > 2\gamma$ {and $x \log (\frac{x}{x-\alpha})$ is decreasing in $x$ \remove{(Comment: is $\alpha$ positive in $(a)$ Reply:The function is decreasing for both positive and negative $\alpha$. Positive $\alpha$ needed for $l\leq m$, negative $\alpha$ for $l > m$.)}}; $(b)$ follows by making $\epsilon$ arbitrarily small; and (c) follows from Lemma~\ref{first_ineq}. 
Then, from \eqref{expectation_likelihood} we have
\begin{equation*}
\mathbb{E}_{\mathcal{P}}[L_{\sigma}] \leq 2\times \mathbb{E}_{\mathcal{P}}[N_l(\sigma)]\times d(p_l||\gamma).  
\end{equation*}
On the other hand, it follows from \cite[Lemma 19]{DBLP:journals/jmlr/KaufmannCG16} that since the estimator $\mathcal{A}$ can correctly recover the set of arms $\mathcal{S}^{\gamma}_{\mathcal{P}}$ with probability at least $(1-\delta)$,
$ \mathbb{E}_{\mathcal{P}}[L_{\sigma}] \geq \log \Bigl(\frac{1}{2.4\delta}\Bigr).$

Combining the two inequalities above and recalling the assumption that $l \le m$, we have
$$ \mathbb{E}_{\mathcal{P}}[N_l(\sigma)] \geq \frac{\log(\frac{1}{2.4\delta})}{2d(p_l||\gamma)} \ \forall \ l \leq m.$$
Using similar arguments, we can also show that 
$$ \mathbb{E}_{\mathcal{P}}[N_l(\sigma)] \geq \frac{\log(\frac{1}{2.4\delta})}{2d(p_l||\gamma)} \ \forall \ l > m.$$

Hence, we have the following lower bound on the query complexity of any $\delta$-true $\gamma$-threshold estimator under this setting: 
$$\mathbb{E}_{\mathcal{P}}[N]\geq \sum_{l=1}^{k} \mathbb{E}_{\mathcal{P}}[N_l(\sigma)] > \sum_{l=1}^{k} \frac{\log(\frac{1}{2.4\delta})}{2d(p_l||\gamma)}.$$
\end{proof}

\subsection{Proof of Lemma \ref{correctness_noisy}}
Let us now define the variable $T_0$ which represented the number of rounds in the first phase of Algorithm \ref{QM2n_alg}. To define $S_0$, we use the following new variables $c_1= (\log 2 + 1) $; $c_2 = \frac{3} {4} (1 - 2p_e)^2 $ ; ${k_1}'=  \frac{e}{2\pi}$ ; ${k_2}'= \frac{1}{2} \exp (2(1-2p_e)^2)$. Additionally, let $c$ denote the largest root of the equation $e^{x(1-2p_e)^2} = x$. Then, $T_0$ is defined\footnote{If the expression in equation \eqref{T0_defn} is not integer, choose the smallest integer greater than or equal to it.} as follows: 
\begin{align}
 T_0 = \frac{4}{\gamma} \max \Bigg\{ c, \frac{2c_1}{c_2} ,\frac{c_1}{2.c_2} +  
 &\sqrt{ (1/c_2)(e/(e-1)) (\log \Bigl(\frac{16k{k_1}'}{c_2 \delta}\Bigr)+ \Bigl(\frac{c_1^2}{4c_2}\Bigr)\Bigr) },\nonumber \\ 
  &\hspace{0.75in}\frac{e}{e-1}\frac{\log\Bigl(\frac{{k_2}'}{ ((1-2p_e)^2)}\sqrt{\frac{16k}{\delta}} \Bigr)}{(1-2p_e)^2},
 \Bigl(1+\frac{33\log (\frac{16k}{\delta})}{(1-2p_e)^2}\Bigr) \Bigg\} .{\label{T0_defn}}
\end{align}
The above definition for $T_0$ is a result of the various constraints that come up while proving the lemmas below. Also, recall that we define $S_0 = \gamma T_0 / 4$. We use the following lemmas to prove Lemma \ref{correctness_noisy}.
\begin{lemma}{\label{lemma_extract}}
The set of extracted bins from the graph $\mathcal{G}$ at the end of the first phase, denoted by $\mathcal{C}'(T_0)$, satisfies the following properties with probability at least $\Bigl(1-\frac{5\delta}{16}\Bigr)$.
\begin{enumerate}
	\item No bin contains samples corresponding to two different support elements.
	\item All the support elements which have at least $2S_0$ samples corresponding to it in the in the first phase have an extracted bin representing it. \remove{\color{red} the bin contains all the samples denoting a support element.} 
	\end{enumerate}
\end{lemma}
\remove{\color{red}This lemma follows from the ideas used in \cite{mazumdar2017clustering, 8732224}.                                                                     
We move the proof the above lemma to supplementary material.}
The proof of this lemma follows from ideas presented in \cite{mazumdar2017clustering, 8732224}. We move the proof to Section~\ref{lemma_extract_proof}.

\begin{lemma}\label{all_occur}
With probability at least $(1 - \frac{\delta}{16})$, for each support element $i$ which belongs to $\mathcal{S}^{\gamma}_{\mathcal{P}}$, and thus has a  probability value above $\gamma$, at least $2S_0$ corresponding samples with value $i$ would have been seen by the end of the first phase after $T_0$ rounds.
\end{lemma}

\begin{proof}
We use the following inequality in our proof which follows from appendix of \cite{Concentration_Inequalities}.

\begin{equation}{\label{ ineq_2}}
\mathbb{P}(X<(1-\epsilon)\mathbb{E}[X])< \exp\Bigl(-\frac{\epsilon^2}{3} \mathbb{E}[X]\Bigr)  \text{  } 0<\epsilon<1.\end{equation}
where $ X = \sum_{i} X_i$ such that $\{X_i\}$ is a set of i.i.d. random variables

For any support element $i$ in $\mathcal{S}^{\gamma}_{\mathcal{P}}$, the number of samples, say $N^i_{T_0}$, seen by the end of $T_0$ rounds with value $i$ satisfies the following:
	%
	\begin{align*}
	\mathbb{P}\Big(N^i_{T_0} \le 2S_0\Big)  = \mathbb{P}\Big(N^i_{T_0} \le \frac{\gamma T_0}{2}\Big)\overset{(a)}{\leq}  \mathbb{P}\Big(N^i_{T_0} \le \frac{\mathbb{E}[N^i_{T_0}]}{2}\Big) &\overset{(b)}{\leq} \exp(-\mathbb{E}[N^i_{T_0}] / 12)\\&\overset{(c)}{\leq}  \exp(-\gamma T_0 / 12)\overset{(d)}{\leq}  \delta/16k.
	\end{align*}
	where $(a)$ and $(c)$ follow from $\mathbb{E}[N^i_{T_0}] \ge \gamma T_0$ for any $i \in \mathcal{S}^{\gamma}_{\mathcal{P}}$; $(b)$ follows from \eqref{ ineq_2} by substituting $\epsilon = \frac{1}{2}$ \remove{\color{red}which version of the bound is this? refer to it as discussed before Comment - SS : updated}and $(d)$ follows since from equation \eqref{T0_defn}.
	\begin{align*}
	\frac{\gamma T_0}{12}  \geq \ \frac{1}{3}\cdot \Bigl(1+\frac{33\log (\frac{16k}{\delta})}{(1-2p_e)^2}\Bigr)   > \ {\log \Bigl(\frac{16k}{\delta}\Bigr)}.
	\end{align*}
	The lemma follows by taking the union bound over all the support elements in $\mathcal{S}^{\gamma}_{\mathcal{P}}$.
\end{proof}
Now we restate and prove Lemma \ref{correctness_noisy}.
\begin{lemma*}
	Given the choice of $\beta^t = \log (\frac{4k(t-T_0)^2}{\delta})$ for each $t > T_0$, where $T_0$ is as defined in \eqref{T0_defn}, Algorithm \ref{QM2n_alg} is a $\delta$-true $\gamma$-threshold estimator under \QMTWONOISY.
\end{lemma*}
\begin{proof}
Let $\mathcal{E}$ denote the event that Algorithm \ref{QM2n_alg} correctly identifies the support elements in $\mathcal{S}^{\gamma}_{\mathcal{P}}$. Also, let $\mathcal{E}_1$ and $\mathcal{E}_2$ be the events that the properties in Lemmas~\ref{lemma_extract} and \ref{all_occur} respectively are satisfied. Then, we have 
$$\mathbb{P}[\mathcal{E}^c] \leq \mathbb{P}[\mathcal{E}^c|(\mathcal{E}_1 \cap \mathcal{E}_2)] + \mathbb{P}[(\mathcal{E}_1 \cap \mathcal{E}_2)^c].$$
From Lemmas~\ref{lemma_extract} and \ref{all_occur}, we have
$$ \mathbb{P}[(\mathcal{E}_1 \cap \mathcal{E}_2)^c] = \mathbb{P}[\mathcal{E}_1^c \cup \mathcal{E}_2^c] \le \frac{\delta}{2}.$$
What remains is to show that $\mathbb{P}[\mathcal{E}^c|(\mathcal{E}_1 \cap \mathcal{E}_2)]  \le \delta / 2$. Henceforth, assume that the events $\mathcal{E}_1$ and $\mathcal{E}_2$ hold true. Note that this implies that when bins are extracted at the end of the first phase of Algorithm \ref{QM2n_alg} after $T_0$ rounds, there will be a unique bin corresponding to each element in $\mathcal{S}^{\gamma}_{\mathcal{P}}$. There might be additional bins corresponding to other support elements as well. 

Next, we consider the second phase of the algorithm where the goal is to identify the bins corresponding to support elements in $\mathcal{S}^{\gamma}_{\mathcal{P}}$. As done in the second phase of Algorithm \ref{QM2_alg} for  \QMTWO, for each round $t > T_0$, we compare the $t^{th}$ sample with a fixed representative element chosen from each bin belonging to a subset $\mathcal{C}'(t)$. Again similar to Algorithm \ref{QM2_alg}, confidence intervals are maintained for each bin and eventually those bins for which the LCB becomes larger than a threshold are identified as the ones corresponding to support elements in $\mathcal{S}^{\gamma}_{\mathcal{P}}$.

The only differences between the second phases of Algorithms \ref{QM2_alg} and \ref{QM2n_alg} are in the values of the empirical estimates and the confidence intervals associated with each bin as well as the value of the threshold against which they are compared. Recall that for $t > T_0$ and a bin $j$ representing support element $i$, $\tilde{\rho}_j^t$ defined in equation \eqref{eq_noisy} represents the fraction of samples since the beginning of the second phase for which the oracle provided a positive response when queried with the representative index from bin $j$. Note that $\tilde{\rho}_j^t$ is a running average of a sequence of i.i.d Bernoulli random variables, each with expected value ${p_i}'$= $p_i\times(1-p_e) + p_e\times(1-p_i) = (1-2p_e)\times p_i + p_e$. Thus as before, Lemma \ref{lemma_cb} applies and can be used to devise the confidence bounds $\mathcal{L}_j(t)$ and $\mathcal{U}_j(t)$ for ${p_i}'$. Finally, the modified threshold is given by $\gamma '$ = $(1-2p_e)\gamma + p_e$ and the bins whose LCB becomes larger than $\gamma '$ will be classified as corresponding to elements from $\mathcal{S}^{\gamma}_{\mathcal{P}}$.

Given the similarities of the two schemes, the arguments for proving the correctness of the above scheme run exactly parallel to the ones made in Lemma \ref{Correctness_QM2} for Algorithm \ref{QM2_alg} and we skip them here for brevity.
\end{proof}
\subsection{Proof of Theorem \ref{QM2_noisyub}}
In this section, we define ${a_i}'$ as $d^{*}({p_i}',\gamma')/2$ and $b$ = $\frac{1}{2} \cdot \log(\frac{4k}{\delta})$. We start with the following lemma.
\begin{lemma}\label{queries with each box noisy}
Assume that the properties in Lemmas~\ref{lemma_extract} and \ref{all_occur} are satisfied. Then, the total number of queries with the bin representing support element $i$ in the second phase of Algorithm \ref{QM2n_alg} is upper bounded by $\frac{e(b-\log({a_i}'))}{(e-1){a_i}'}$  with probability at least $(1-\delta/4k)$.
\end{lemma}
\begin{proof}
Since the properties in Lemmas~\ref{lemma_extract} and \ref{all_occur} are satisfied, every bin at the end of the first phase represents a different support element. Let bin $j$ denote support element $i$. Recall from equation \eqref{eq_noisy} that  $\tilde{\rho}_j^t$ denotes the fraction of samples seen in the second phase till round $t$ that receive a positive response when compared to the representative element in Bin $j$. We have $\mathbb{E}[\tilde{\rho}_j^t] = {p_i}' =  p_i\times (1-p_e) + p_e\times(1-p_i) = (1-2p_e)\times p_i + p_e$ and $\mathcal{L}_j(t)$ and $\mathcal{U}_j{(t)}$ denote the lower and upper confidence bounds of bin $j$ respectively, as defined in equations \eqref{eqn_noisylb} and \eqref{eqn_noisyub} respectively. Accordingly in equation \eqref{eqn:c2tnoise}, the LCB and UCB of bin $j$ are compared with a modified threshold given by $\gamma' =   (1-2p_e)\times \gamma + p_e$ to decide how long it will be retained in the subset $\mathcal{C}'(t)$ of bins that new samples are compared against. 

The above setting of the second phase is similar to the \QMONE \ model where each new sample with index above $T_0$ would fall in the bin representing support element $i$ with probability ${p_i}'$. Thus similar results apply and in particular, Lemmas~\ref{lemma_exittime} and \ref{lemma_timeupper} can be used to show that the total number of queries with a bin representing support element $i$ is upper bounded by $\frac{e(b-\log({a_i}'))}{(e-1){a_i}'}$ with probability at least $(1-\frac{\delta}{4k})$.
\remove{
This proof follows along the same lines as that of Lemmas \ref{lemma_exittime} and  \ref{lemma_timeupper}. This is because, 

Recall that $\mathbb{E}[\mathcal{Z}_j^t]= {p_i}' =  p_i\times (1-p_e) + p_e\times(1-p_i) = (1-2p_e)\times p_i + p_e$ for $t > T_0$ where indices less than $T_0$ in Bin $j$ denotes support element $i$.

Thus the classifier would be $\gamma' =   (1-2p_e)\times \gamma + p_e$ for $t > T_0$ to classify the bins which have the probability of the support elements of indices below $T_0$ above $\gamma$.

Now the setting is similar as that in \QMONE \ model where each new sample with index above $T_0$ would fall in bin with support element $i$ with probability ${p_i}'$. Thus the same theorems and lemmas as that in \ref{lemma_exittime} and \ref{lemma_timeupper} can be used to show that the total number of queries with a bin which denoted support element $i$ at Round $T_0$ is upper bounded by $\frac{e(b-\log({a_i}'))}{(e-1){a_i}'}$ with probability at least $(1-\frac{\delta}{4k})$.
}
\end{proof}
Now we restate and prove Theorem \ref{QM2_noisyub}.
\begin{theorem*}
	Let $\mathcal{A}$ denote the estimator in Algorithm~\ref{QM2n_alg} with  $\beta^t = \log (\frac{4k(t-T_0)^2}{\delta})$ for each $t > T_0$ where $T_0$ is as defined in \eqref{T0_defn}. Let $Q_{\delta,\gamma}^{\mathcal{P}}(\mathcal{A})$ be the corresponding query complexity for a given distribution $\mathcal{P}$ under \QMTWONOISY \ and define $q = \min\{T_0,k\}$, ${p_i}'= (1-2p_e)\times p_i+p_e$ and $\gamma' = (1-2p_e)\gamma+p_e$. Then we have $$Q_{\delta,\gamma}^{\mathcal{P}}(\mathcal{A}) \leq \sum_{i=1}^{q} \frac{2e.\log \Bigl(\sqrt{\frac{4k}{\delta}}\frac{2}{d^{*}({{p}_i}',\gamma')}\Bigr)}{(e-1).d^{*}({p_i}',\gamma')}  + \frac{T_0(T_0-1)}{2} $$
with probability at least $(1- 2\delta)$.
\end{theorem*}
\begin{proof}
Let us first bound the total number of queries in the second phase, i.e., post round $T_0$. Since the bins are created only at the end of the first phase, the total number of bins must be upper bounded by $T_0$. Furthermore, if the properties in Lemmas~\ref{lemma_extract} and \ref{all_occur} are satisfied, there is at most one bin corresponding to each support element which implies the total number of bins must be upper bounded by $q= \min\{k,T_0\}$.

Assuming that the properties in Lemmas~\ref{lemma_extract} and \ref{all_occur} are satisfied, from Lemma \ref{queries with each box noisy}, we have that the total number of queries with the bin representing support element $i$ is upper bounded by $\frac{e(b-\log({a_i}'))}{(e-1){a_i}'}$ with probability at least $(1-\frac{\delta}{4k})$. Taking the union bound over all the $q$ bins we can say with probability at least $(1-\delta)$ that total number of queries in the second phase is bounded by $\sum_{i=1}^{q} \frac{2e.\log \Bigl(\sqrt{\frac{4k}{\delta}}\frac{2}{d^{*}({{p}_i}',\gamma')}\Bigr)}{(e-1).d^{*}({p_i}',\gamma')}$. 

Now since the probability that the properties in Lemmas~\ref{lemma_extract} and \ref{all_occur} are all satisfied is at least $(1 - \delta)$, the total number of queries in the second phase is upper bounded by $\sum_{i=1}^{q} \frac{2e.\log \Bigl(\sqrt{\frac{4k}{\delta}}\frac{2}{d^{*}({{p}_i}',\gamma')}\Bigr)}{(e-1).d^{*}({p_i}',\gamma')}$ with probability at least $(1-2\delta)$. Finally, noting that there are exactly $\frac{T_0(T_0-1)}{2}$ queries in the first phase of the algorithm,  our proof is complete.    
\end{proof}

\subsection{Proof of Lemma \ref{lemma_extract}}{\label{lemma_extract_proof}}
We use the following claims  to prove Lemma \ref{lemma_extract}. These claims and their proofs are very similar to those in \cite{mazumdar2017clustering, 8732224}. Recall the following terms defined in main paper.\\
$c_1= (\log 2 + 1) $; $c_2 = \frac{3} {4} (1 - 2p_e)^2 $ ; ${k_1}'=  \frac{e}{2\pi}$ ; ${k_2}'= \frac{1}{2} \exp (2(1-2p_e)^2)$. Additionally, let $c$ denote the largest root of the equation $e^{x(1-2p_e)^2} = x$.
We now define $K'$ as follows.

\begin{align}\label{K_prime_defn}
K'= \max \Bigl\{c,\frac{2c_1}{c_2} ,\frac{c_1}{2.c_2} 
+  \sqrt{ (1/c_2)(e/(e-1)) \Bigl(\log \Bigl(\frac{2k{k_1}'}{c_2 \delta}\Bigr)+ \Bigl(\frac{c_1^2}{4c_2}\Bigr)\Bigr) }, \nonumber\\
 \frac{e}{e-1}\frac{\log\Bigl(\frac{{k_2}'}{ ((1-2p_e)^2)}\sqrt{\frac{2k}{\delta}} \Bigr)}{(1-2p_e)^2}\Bigr\}.
\end{align}

\begin{claim}\label{comp_subcluster}
Consider a graph $\mathcal{G}(\hat{V},\hat{E})$ where $|\hat{V}| \geq K'$ defined in \eqref{K_prime_defn} and edge weights are i.i.d. random variables taking the value $-1$ with probability $p_e$ $< \frac{1}{2}$ and $1$ with probability
$(1 - p_e)$. Then, $wt(\mathcal{G}) > wt(\mathcal{G}')$ for any subgraph $\mathcal{G}' \subset \mathcal{G}$, i.e.,
the MWS extracted from $\mathcal{G}$ will include the entire node set $\hat{V}$ with probability at least $(1- \frac{\delta}{k})$.
\end{claim}

\begin{proof}

We will use the following inequalities in the proof below.
\begin{equation}\label{bound1}
\mbox{Stirling's inequality \cite{feller1}}:      \sqrt{2\pi} n^{n+\frac{1}{2}}.e^{-n} < n! < e.n^{n+\frac{1}{2}}.e^{-n} \ \forall \ n \ge 1.
\end{equation}
\begin{equation}\label{bound2}
\mbox{\cite[Appendix]{Concentration_Inequalities}} :  \   \mathbb{P}(X \leq \mathbb{E}[X]-t) \leq
    \exp\Bigl(-\frac{t^2}{2n}\Bigr),
    \end{equation}
    \begin{equation}
         \text{where } t > 0,   X = \sum_{i=1}^{n}X_i \text{ such that } \{X_i\}  \text{is a set of i.i.d. radom varibles. } \nonumber
\end{equation}
%
%
Let $S$ be a subset of $\hat{V}$ and we try to compute the probability that the  $wt(S)$ is greater than $wt(\hat{V})$. Let us denote the weights of the edge between node $i$ and node $j$ as $w_{ij}$. Then
\begin{align*}
 \mathbb{P} \Big(\sum_{i,j \in \hat{V};i \neq j} w_{ij} < \sum_{i,j \in S; i \neq j ;S \subseteq \hat{V}} w_{ij} \Big) 
  =& \mathbb{P} \Big(\sum_{(i,j) \in (\hat{V},\hat{V});(i,j) \notin (S,S);i \neq j} w_{ij} < 0 \Big)\\
  \overset{(a)}{\leq} & \exp\Big( -2 {(1-2p_e)}^2 \Big[  { |\hat{V}| \choose 2} - { S \choose 2} \Big] \Big).
\end{align*}
Note that $(a)$ follows from \eqref{bound2}.\\ 
Applying the union bound gives us,

\begin{align*} 
     \hspace{-10.0 em}\mathbb{P}(\text{MWS} \neq \hat{V})
     \leq &\sum_{|S|=1}^{|\hat{V}|-1}  { |\hat{V}| \choose |S|} \mathbb{P}\Big(\sum_{(i,j) \in (\hat{V},\hat{V});(i,j) \notin (S,S);i \neq j} w_{ij} < 0 \Big)\\
     \leq &\sum_{|S|=1}^{|\hat{V}|-1}  { |\hat{V}| \choose |S|} \exp\Big( -2 {(1-2p_e)}^2 \Big[  { |\hat{V}| \choose 2} - { |S| \choose 2} \Big] \Big) \\ 
     = & \sum_{|S|=1}^{\frac{|\hat{V}|}{2}}  { |\hat{V}| \choose |S|} \exp\Big( -2 {(1-2p_e)}^2 \Big[  { |\hat{V}| \choose 2} - { |S| \choose 2} \Big] \Big) \\
     &   +  \sum_{|S|= \frac{|\hat{V}|}{2}+1}^{|\hat{V}|-1}  { |\hat{V}| \choose |S|} \exp\Big( -2 {(1-2p_e)}^2 \Big[  { |\hat{V}| \choose 2} - { |S| \choose 2} \Big] \Big) \\
     \overset{(b)}{\leq} &  \sum_{|S|=1}^{\frac{|\hat{V}|}{2}}  { |\hat{V}| \choose \frac{|\hat{V}|}{2}} \exp\Bigl(-(1-2p_e)^2 \Bigl( \frac{3|\hat{V}|^2}{4} - \frac {|\hat{V}|}{2} \Bigr)\Bigr)\\
     &+ \sum_{|S|=\frac{|\hat{V}|}{2}+1}^{|\hat{V}|}  { |\hat{V}| \choose {|\hat{V}|-1}} \exp(-2(1-2p_e)^2 \Bigl( |\hat{V}| - 1 \Bigr)  \Bigr)\\       
     \hspace{-0em} \overset{(c)}{\leq} & \frac{|\hat{V}|}{2} \frac{e}{\pi} \frac{2^{| \hat{V}|}}{\sqrt{|\hat{V}|}} \exp( \Bigl(-(1-2p_e)^2 \Bigl( \frac{3|\hat{V}|^2}{4} - \frac {|\hat{V}|}{2} \Bigr)\Bigr)  \Bigr)\\ &+{k_2}' |\hat{V}|^2 \exp(-2(1-2p_e)^2(|\hat{V}|)) \\
     \overset{(d)} {\leq} & {k_1}' \sqrt {|\hat{V}|} \exp(|\hat{V}| (\log 2 + 1) - (1 - 2p_e)^2 . \frac{3} {4} {|\hat{V}|^2}))\\ &+{k_2}' |\hat{V}|^2 \exp(-2(1-2p_e)^2(|\hat{V}|))\\
     \overset{(e)}{\leq}& \frac{\delta}{2k} +
     \frac{\delta}{2k}
     \leq  \frac{\delta}{k}.     
\end{align*}

The inequality for the first term in $(b)$ follows since $\sum_{|S|=1}^{\frac{|\hat{V}|}{2}}  { |\hat{V}| \choose |S|} \exp\Big( -2 {(1-2p_e)}^2 \Big[  { |\hat{V}| \choose 2} - { |S| \choose 2} \Big] \Big)$ takes maximum value at $|S| = \frac{|\hat{V}|}{2}$. The inequality at second term in $(b)$ can be shown by arguing that term $\sum_{|S|=\frac{|\hat{V}|}{2}+1}^{|\hat{V}|}  { |\hat{V}| \choose {|\hat{V}|-1}} \exp(-2(1-2p_e)^2 \Bigl( |\hat{V}| - 1 \Bigr)  \Bigr) $ takes maximum value at $S = |\hat{V}| - 1$ which we show below. 

Consider the function 
\begin{align*}
f(a) 
=  &{ |\hat{V}| \choose {|\hat{V}|-a}} \exp\Biggl(-2(1-2p_e)^2 \Biggl(  {|\hat{V}| \choose 2} - {(|\hat{V}|-a) \choose 2} \Biggr)  \Biggr)\\
= & { |\hat{V}| \choose {|\hat{V}|-a}} \exp((-(1-2p_e)^2) (2a|\hat{V}| - a^2+a)).
\end{align*}

We wish to show that the maximum of $f(a)$ occurs at $a = 1$. 

After the requisite algebraic simplification, 
\begin{align*}
    \frac{f(a)}{f(a+1)} 
    =  \frac{a+1}{|\hat{V}|-a} \frac{e^{2(1-2p_e)^2|\hat{V}|}}{e^{(1-2p_e)^2(2a+2)}}
    \overset{(f)}{\geq}  \frac{2}{|\hat{V}|-1} \frac{e^{2(1-2p_e)^2|\hat{V}|}} {e^{(1-2p_e)^2(2.(\frac{|\hat{V}|}{2}-1)+2)}}
    \geq  \frac{2}{|\hat{V}|} e^{(1-2p_e)^2|\hat{V}|}
    \overset{(g)}{\geq}   1.
\end{align*}

Note that $(f)$ follows on minimising each fraction individually. $(g)$ follows from the fact that $c$ is the largest solution of $e^{(1-2p_e)^2x} = x$ which implies that $e^{(1-2p_e)^2|\hat{V}|} > |\hat{V}|$ for  $|\hat{V}| > c$. Therefore, $f(a)$ is a decreasing function of $a$ and takes maximum value at $a=1$.


The inequality for first term in $(c)$ follows on applying \eqref{bound1} whereas for the second term follows on applying ${k_2}'$ = $(1/2)\exp(2(1-2p_e)^2)$. The inequality at $(d)$ follows by bounding $\exp((1-2p_e)^2\frac{|\hat{V}|}{2})$ with $\exp(|\hat{V}|)$.

Now let us prove the inequality on first term in $(e)$ by proving 
\begin{equation*}
{k_1}' \sqrt {|\hat{V}|} \exp\bigg(|\hat{V}| (\log 2 + 1) - (1 - 2p_e)^2 . \frac{3} {4} {|\hat{V}|^2}\bigg) \leq \frac{\delta}{2k}.
\end{equation*}

Since $|\hat{V}| > 2.\frac{c_1}{c_2}$ and $\frac{c_1}{c_2}>1$,  $\sqrt {|\hat{V}|} < \Bigl(|\hat{V}|-\frac{c_1}{2.c_2}\Bigr)^2$.

Now consider  the first term in inequality in $(d)$.
\begin{align*}
     {k_1}' \sqrt {|\hat{V}|} \exp\bigg(|\hat{V}|& (\log 2 + 1)- (1 - 2p_e)^2 . \frac{3} {4} {|\hat{V}|^2}\bigg)
    =  {k_1}' \sqrt {|\hat{V}|} \exp(c_1|\hat{V}|  - c_2 {|\hat{V}|^2})\\
    \overset{(j)}{\leq} & {k_1}' \Bigl(|\hat{V}|-\frac{c_1}{2.c_2}\Bigr)^2 \exp\Bigl(\frac{c_1^2}{4c_2}\Bigr) \exp\Bigl(-c_2 \Bigl(|\hat{V}|-\frac{c_1}{2.c_2}\Bigr)^2 \Bigr)
    \overset{(l)}{\leq}  \frac{\delta}{2}
\end{align*}

$(j)$ follows from $\sqrt {|\hat{V}|} < \Bigl(|\hat{V}|-\frac{c_1}{2.c_2}\Bigr)^2$.


Let us now prove the inequality in $(l)$.

We can say that  
\begin{align*}
   &  \hspace{-0em} |\hat{V}| > \frac{c_1}{2.c_2} + \sqrt{\frac{e}{c_2(e-1)} \bigg(\log \Bigl(\frac{2k{k_1}'}{c_2 \delta}\Bigr)+ \Bigl(\frac{c_1^2}{4c_2}\Bigr)\bigg)} \\
   &  \hspace{-0em} \Rightarrow  - c_2\Bigl(|\hat{V}| - \frac{c_1}{2.c_2}\Bigr)^2 < -\frac{e}{e-1} \bigg(\log \Bigl(\frac{2k{k_1}'}{c_2 \delta}\Bigr)+ \Bigl(\frac{c_1^2}{4c_2}\Bigr)\bigg)\\
   &  \hspace{-0em} \overset{(h)}{\Rightarrow}  -c_2\Bigl(|\hat{V}|-\frac{c_1}{2.c_2}\Bigr)^2 < W_{-1} \Bigl(\frac{c_2.\delta}{2k{k_1}'}\exp\Bigl(-\frac{c_1^2}{4c_2}\Bigr)\Bigr)\\
   & \hspace{-0em}  \overset{(i)}{\Rightarrow}  -c_2 \Bigl(|\hat{V}|-\frac{c_1}{2.c_2}\Bigr)^2 \exp \Bigl(-c_2 \Bigl(|\hat{V}|-\frac{c_1}{2.c_2}\Bigr)^2 \Bigr) >  -c_2 \frac{\delta}{2k{k_1}'}\exp\Bigl(-\frac{c_1^2}{4c_2}\Bigr)\\
   & \hspace{-0em} \Rightarrow {k_1}' \Bigl(|\hat{V}|-\frac{c_1}{2.c_2}\Bigr)^2 \exp\Bigl(\frac{c_1^2}{4c_2}\Bigr) \exp\Bigl(-c_2 \Bigl(|\hat{V}|-\frac{c_1}{2.c_2}\Bigr)^2 \Bigr) < \frac{\delta}{2k}.
\end{align*}


Recall that $W_{-1}$ is the lower root in lambert function as defined in the proof of Theorem 1 of our paper. Now using Theorem 3.1 of \cite{Lambert}, we say that $W_{-1}(-\frac{c_2.\delta}{2k{k_1}'}.\exp(-\frac{c_1^2}{4c_2}) )$ is lower bounded by $-(e/(e-1)) (\log (\frac{2k{k_1}'}{c_2 \delta})+ (\frac{c_1^2}{4c_2}))$. This would in turn imply the implication in $(h)$.
Note that the implication in $(i)$ follows from the the fact that $W_{-1}$ is the lower root of the lambert function. 



Thus for $|\hat{V}| > \max\Bigl(\frac{c_1}{2.c_2} + \sqrt{ (1/c_2)(e/(e-1)) (\log (\frac{2k{k_1}'}{c_2 \delta})+ (\frac{c_1^2}{4c_2})) },\frac{2.c_1}{c_2},c\Bigr) $, the inequality in $(l)$ is proven. Therefore,  the first term in $(e)$ is upper bounded by $\frac{\delta}{2k}$.

Now consider the second term on the inequality in $(e)$.

We can say the following:
\begin{align*}
    & |\hat{V}| > \frac{e}{e-1}\frac{\log\Bigl(\frac{{k_2}'}{ (1-2p_e)^2}\sqrt{\frac{2k}{\delta}} \Bigr)}{(1-2p_e)^2}\\
     \Rightarrow & -(1-2p_e)^2 |\hat{V}| < \frac{e}{e-1}\log\Bigl(\frac{{k_2}'}{ (1-2p_e)^2}\sqrt{\frac{2k}{\delta}} \Bigr)\\
     \overset{(m)}{\Rightarrow} & -(1-2p_e)^2 |\hat{V}| < W_{-1}\Bigl(-\frac{((1-2p_e)^2)}{{k_2}'}\sqrt{\frac{\delta}{2k}}  \Bigr)\\
      \overset{(n)}{\Rightarrow} &  (-(1-2p_e)^2|\hat{V}| \exp(-(1-2p_e)^2|\hat{V}|) > -\frac{((1-2p_e)^2)}{{k_2}'}\sqrt{\frac{\delta}{2k}} \\
     \Rightarrow & |\hat{V}| \exp(-(1-2p_e)^2|\hat{V}|`) < \frac{1}{{k_2}'}\sqrt{\frac{\delta}{2k}}\\
     \Rightarrow & {k_2}' |\hat{V}|^2 \exp(-2(1-2p_e)^2(|\hat{V}|)) < \frac{\delta}{2k}.
\end{align*}
Note that the implication in $(m)$ follows from theorem 3.1 of \cite{Lambert} which implies that the value of $W_{-1}\Bigl((\frac{((1-2p_e)^2)}{{k_2}'}\sqrt{\frac{\delta}{2k}} \Bigr) \Bigr)$ is lower bounded by $\frac{e}{e-1}\log\Bigl(\frac{{k_2}'}{ ((1-2p_e)^2)}\sqrt{\frac{2k}{\delta}} \Bigr)$. The implication in $(n)$ follows from the definition of $W_{-1}$ similar to the reasoning in $(i)$. Thus, the second inequality in $(e)$ is proven implying that the claim is also proven. 
\end{proof}

Now we state and prove the next claim which would be used to prove Lemma \ref{lemma_extract}.
~
\begin{claim}\label{no_overlap}
Consider a graph $G'$ whose vertices are partitioned into multiple clusters. The weight of edges between any pair of nodes in the same cluster are random variables which take value $-1$ with probability $p_e$ and $1$ with probability $(1-p_e)$. Also weight of edges between two nodes which does not lie in the same cluster takes value $1$ with probability $p_e$ and $-1$ with probability $(1-p_e)$. We assume that all the weights of the edges are independent.\\
Let S be the MWS of $G'$. If $ |S| >
 \Bigl(1+\frac{33\log (\frac{2}{\delta})}{(1-2p_e)^2}\Bigr) $ then we can say with probability at least $(1-\delta)$ that it can not contain nodes from multiple sub-clusters. 
\end{claim}
\begin{proof}
Let S be a sub-component contain nodes from at least two sub-clusters. Let the clusters be denoted by $V_i$ and we denote $C_i=S \cap V_i$ and $j* = \argmin_{i:C_i \neq \phi} |C_i|$.

Let the weight of the edge in the graph between the nodes $i$ and $j$ be denoted by $w_{ij}$. We claim that $$\sum_{i,j \in S,i<j} w_{ij} < \sum_{i.j \in S \textbackslash C_{j*} , i<j } w_{ij} $$ with probability at least $(1-\delta)$. The above condition is equivalent to 
$$\sum_{i,j \in C_{j*},i<j} w_{ij} + \sum_{i \in S, j \in S \textbackslash C_{j*} } w_{ij} < 0$$

We use the following equations from appendix of \cite{Concentration_Inequalities} in the proof. 
\begin{equation}{\label{ ineq_1}}
\mathbb{P}(X>(1+\epsilon)\mathbb{E}[X])< \exp\Bigl(-\frac{\epsilon^2}{3} \mathbb{E}[X]\Bigr) \text{  } 0<\epsilon<1.
\end{equation}
\begin{equation}{\label{ ineq_2}}
\mathbb{P}(X<(1-\epsilon)\mathbb{E}[X])< \exp\Bigl(-\frac{\epsilon^2}{3} \mathbb{E}[X]\Bigr)  \text{  } 0<\epsilon<1.\end{equation}
where $ X = \sum_{i} X_i$ such that $\{X_i\}$ is a set of i.i.d. random variables. We divide the proof into two cases.

\textbf {Case 1:}
$|C_{j*}| > \sqrt{\frac{108}{1-2p_e} \log (\frac{2}{\delta})}$\\
Now 
\begin{align*}
 \mathbb{P}\Bigg( \sum_{i,j \in C_{j*},i<j} w_{ij}> \Big(1+\frac{1}{3}\Big) (1-2p_e) {|C_{j*}| \choose 2} \Bigg)
 \overset{(a)}{\leq} & \exp\Bigg(-(1/3)^3 (1-2p_e) {|C_{j*}| \choose 2}\Bigg)\\
 < & \exp\bigg(-(1/3)^3 (1-2p_e) \frac{|C_{j*}|^2}{4} \bigg)
 \leq  \delta / 2.
\end{align*}
$(a)$ follows from \eqref{ ineq_1} by putting $\epsilon = 1/3$.
\begin{align*}
 \mathbb{P}\Bigl( \sum_{i \in C_{j*}, j \in S \textbackslash C_{j*}  } w_{ij}> -(1-\frac{1}{3}) (1-2p_e) |C_{j*}|  |S \textbackslash C_{j*}|\Bigr)
 \overset {(b)} {\leq}  & \exp(-(1/3)^2 (1-2p_e) |C_{j*}| |S \textbackslash C_{j*}|)\\
 \overset {(c)} {\leq} & \exp(-(1/3)^3 (1-2p_e) \frac{|C_{j*}|^2}{2}) 
  \leq  \delta /2. 
\end{align*}
The inequality in $(b)$ holds from \eqref{ ineq_1} and the inequality in $(c)$ holds due to  $|S \textbackslash C_{j*}| > |C_{j*}|$ since $C_{j*}$ is the smallest cluster.
Now we apply union bound on the previous two proven events and can say with probability at least $(1-\delta)$.
\begin{align*}
\sum_{i,j \in C_{j*},i<j} w_{ij} + \sum_{i \in C_{j*}, j \in S \textbackslash C_{j*}  } w_{ij}
 \leq &  \Bigl((4/3) (1-2p_e) {|C_{j*}| \choose 2} - (2/3) (1-2p_e) |C_{j*}|  |S \textbackslash C_{j*}|\Bigr) \\
 \leq &\Bigl((4/3) (1-2p_e) \frac{|C_{j*}|^2}{2}   -(2/3) (1-2p_e) |C_{j*}|  |C_{j*}|\Bigr)\leq  0.
\end{align*}

\textbf {Case 2:}
$|C_{j*}|< \sqrt{\frac{108}{1-2p_e} \log (\frac{2}{\delta})}$\\
\begin{align*}
 \mathbb{P}\Bigl( \sum_{i \in C_{j*}, j \in S \textbackslash C_{j*}  } w_{ij}> -(1-\frac{1}{2}) (1-2p_e)& |C_{j*}|  |S \textbackslash C_{j*}|\Bigr)\\
 \leq &\exp(-(1/2)^2 (1/3) (1-2p_e) |C_{j*}| |S \textbackslash C_{j*}|)\\
 \overset {(d)} {\leq} & \exp(-(1/2)^2 (1/3) (1-2p_e) (|S|-1) )
\overset {(e)} {\leq} \delta /2.
\end{align*}
The inequality $(d)$ follows since $|C_{j*}| |S \textbackslash C_{j*}|$ $\geq$  $(|S|-1)$.
The inequality $(e)$ follows since 
\begin{align*}
|S| \geq & \Bigl(1+\frac{33\log (\frac{2}{\delta})}{(1-2p_e)^2}\Bigr) 
\geq   \Bigl(1+\frac{12\log (\frac{2}{\delta})}{(1-2p_e)^2}\Bigr).
\end{align*}
Now take $|C_{j*}| = x$.
Hence, $$\sum_{i,j \in C_{j*},i<j} w_{ij} \leq \frac{x^2}{2}.$$

Thus, we say with probability at least (1- $\delta$) that
\begin{align*}
 \sum_{i,j \in C_{j*},i<j} w_{ij} + \sum_{i \in C_{j*}, j \in S \textbackslash C_{j*}  } w_{ij} 
&\leq  \frac{x^2}{2} - (1-2p_e) |C_{j*}|  |S \textbackslash C_{j*}|/2`\\
 &\leq  x^2 - (1/2) (1-2p_e) x(|S|-x)\\
&\leq  x (\frac{3x}{2} - (1/2) (1-2p_e) |S|)
\overset{(f)}{<}  0.
\end{align*}
The inequality at $(f)$ is true since 
\begin{align*}
|S| > \Bigl(1+\frac{33\log (\frac{2}{\delta})}{(1-2p_e)^2}\Bigr)\geq  \frac{3 \sqrt{\frac{108}{1-2p_e} \log (\frac{2}{\delta})}}{1-2p_e}.
\end{align*}

Thus we say with probability at least $(1-\delta)$ that the MWS contains points from a single cluster.
\end{proof}

\begin{claim}{\label{new_claim}}
	Consider a support element $a$ with less than $S_0$ samples denoting it in the graph. Consider another support element $b$ with at least than $2S_0$ samples denoting it. The probability that the weight of the sub-graph which is a subset of the samples corresponding to the support element $a$ is smaller than the weight of the sub-graph containing all the samples corresponding to the support element $b$ with probability at least $(1-\frac{\delta}{16k})$.
\end{claim}

\begin{proof}
	Denote the nodes corresponding to a support element $a$ as a cluster $c_a$ and denote its subset as $s_{c_a}$. Similarly, we denote the nodes corresponding to a support element $b$ as $c_b$. We denote the weights of edges between nodes $i$ and $j$ as $w_{ij}$. 
	
	We consider the probability that $\sum_{i,j \in s_{c_a}} w_{ij} > \sum_{i,j \in c_b} w_{ij}$. 
	
	\begin{align*}
	\mathbb{P}[\sum_{i,j \in {s_{c_a}}} w_{ij} > \sum_{i,j \in c_b} w_{ij}]
	\leq & \mathbb{P}[\sum_{i,j \in c_b} w_{ij} - \sum_{i,j \in s_{c_a}} w_{ij}<0]\\
	\overset{(a)}{\leq} &\exp\Bigl(-\frac{1}{3} \mathbb{E}\Bigl[\sum_{i,j \in c_b} w_{ij} - \sum_{i,j \in s_{c_a}} w_{ij}\Bigr]\Bigr)\\
	\overset{(b)}{\leq} & \exp(-\frac{1}{6} (1-2p_e) ((|c_b|-1)^2-|s_{c_a}|^2))\\
	\overset{(c)}{\leq} & \exp(-\frac{1}{2} (1-2p_e) {{S_0}^2})
	\overset{(d)}{\leq}  \frac{\delta}{16k} 
	\end{align*}
	
Note that $(a)$ follows from \eqref{ ineq_2} and substituting $\epsilon=1$. $(b)$ follows since $\mathbb{E}\Bigl[\sum_{i,j \in c_b} w_{ij} - \sum_{i,j \in s_{c_a}} w_{ij}\Bigr] = (1-2p_e) (\frac{(|c_b|)(|c_b|-1)}{2}- \frac{(|s_{c_a}|)(|s_{c_a}|-1)}{2}) \leq \frac{1}{2}(1-2p_e) ((|c_b|-1)^2-|s_{c_a}|^2)$.

$(c)$ follows since minimum value of $|c_b|$ is $2S_0$ whereas maximum value of $|s_{c_a}|$ is $(S_0-1)$, thus $((|c_b|-1)^2-|s_{c_a}|^2)$ is lower bounded by $3{S_0}^2$. $(d)$ follows since $S_0 \geq \Bigl(1+\frac{33\log (\frac{16k}{\delta})}{(1-2p_e)^2}\Bigr)$.

\end{proof}

Let us restate and prove Lemma \ref{lemma_extract}.

\begin{lemma*}
	The set of extracted bins from the graph $\mathcal{G}$ at the end of the first phase, denoted by $\mathcal{C}'(T_0)$, satisfies the following properties with probability at least $\Bigl(1-\frac{5\delta}{16}\Bigr)$.
	\begin{enumerate}
		\item No bin contains samples corresponding to two different support elements.
		\item All the support elements which have at least $2S_0$ samples corresponding to it in the first phase have an extracted bin representing it. \remove{\color{red} the bin contains all the samples denoting a support element.} 
	\end{enumerate}
\end{lemma*}

\begin{proof}



Let $\tilde{K}$ be defined by substituting $\frac{\delta}{8}$ for $\delta$ in the expression of $K'$ in \eqref{K_prime_defn}.

By substituting $\frac{\delta}{8k}$ for $\delta$ in Claim \ref{no_overlap}, we can say any extracted Maximum Weighted SubGraph (MWS) of size larger than $\Bigl(1+\frac{33\log (\frac{16k}{\delta})}{(1-2p_e)^2}\Bigr)$ can have indices corresponding to multiple support element with probability at most $\frac{\delta}{8k}$.  

Consider each extracted MWS of size larger than $S_0$. Since, $S_0 > \Bigl(1+\frac{33\log (\frac{16k}{\delta})}{(1-2p_e)^2}\Bigr)$, we say that this MWS extracted has indices representing  only a support element with probability at least $(1-\frac{\delta}{8k})$.

Consider all support elements (denoted by $\tilde{S}_{min}$) with less than $S_0$ indices denoting it. Consider all support elements (denoted by $\tilde{S}_{max}$) with more than $2S_0$ indices denoting it. 

 
We argue that the events in Claim \ref{new_claim} and Claim \ref{comp_subcluster} (substituting $\frac{\delta}{8}$ for $\delta$) and Claim \ref{no_overlap} (substituting $\frac{\delta}{8k}$ for $\delta$) would imply for each MWS extracted, we can say that no MWS representing an element in $\tilde{S}_{min}$ would occur till MWS corresponding to all support elements in $\tilde{S}_{max}$ have been extracted. 

Suppose not. Consider that MWS corresponding to some support element in $\tilde{S}_{min}$ occurs before all elements in $\tilde{S}_{max}$ have occurred in some MWS created in previous rounds. Since there are some other elements in $\tilde{S}_{max}$ not a part of MWS extracted, it would imply that a subset of all indices denoting some support element in $\tilde{S}_{min}$ has higher weight than the sub-graph corresponding to the support element in $\tilde{S}_{max}$ implying event in Claim $\ref{new_claim}$ is violated.

Now the event that no MWS representing an element in $\tilde{S}_{min}$ occurs till MWS corresponding to all support elements in $\tilde{S}_{max}$ have been extracted and events in Claim \ref{new_claim} and Claim \ref{comp_subcluster} (substituting $\frac{\delta}{8}$ for $\delta$) and Claim \ref{no_overlap} (substituting $\frac{\delta}{8k}$ for $\delta$) imply that extraction of MWS does not stop till all elements in $\tilde{S}_{max}$ have been extracted in some MWS. Suppose not. This would imply that we get an MWS of size less than $S_0$ before all elements in $\tilde{S}_{max}$ have been put in some MWS implying the event that some element in $\tilde{S}_{min}$ occurs MWS before all support elements in $\tilde{S}_{max}$ have been extracted as Claim $\ref{comp_subcluster}$ holds true for all support elements of size larger than $S_0$.

Now the probability of the events described in Claim \ref{new_claim} and Claim \ref{comp_subcluster} (substituting $\frac{\delta}{8}$ for $\delta$) and Claim \ref{no_overlap} (substituting $\frac{\delta}{8k}$ for $\delta$) can be union bounded over all support elements to argue that the probability is at least $(1-\frac{5\delta}{16})$.

Thus, we argue that both the events in theorem hold true with probability at least $(1-\frac{5\delta}{16})$.

\end{proof}

\subsection{Proof of Lemma \ref{first_ineq}}{\label{first_ineq_proof}}
Let us restate and prove Lemma \ref{first_ineq}.

\begin{lemma*}
	For all $1 \geq p_t, \gamma \geq 0$, the following inequality holds true.	
	$$p_t \log \Bigl(\frac{p_t}{\gamma}\Bigr)+ (2.\gamma) \log\Bigl(\frac{2.\gamma}{p_t + \gamma}\Bigr) \leq 2 d(p_t||\gamma).$$ 
\end{lemma*}

\begin{proof}

Consider the function $f(\gamma) =  - p_t \log (\frac{p_t}{\gamma-\epsilon})- (2.\gamma) \log(\frac{2.\gamma}{p_t + \gamma}) + 2 d(p_t||\gamma)$.
On differentiating the function with respect to $\gamma$, we have 
$$f'(\gamma) = -\frac{p_t}{\gamma} + \frac{2(1-p_t)}{(1-\gamma)} - 2.\log\Bigl(\frac{2.\gamma}{p_t+\gamma}\Bigr) - \frac{2.p_t}{p_t+\gamma}.$$
Note that the $f'(\gamma) = f(\gamma) = 0$ for $\gamma = p
_t$.
Now on double differentiating $f(\gamma)$, we have 
\begin{align*}
f''(\gamma) = & \frac{p_t}{\gamma^2} + \frac{2(1-p_t)}{(1-\gamma)^2} - \frac{2.{p_t}^2}{\gamma{(p_t+\gamma)}^2} 
=  \frac{2(1-p_t)}{(1-\gamma)^2} + \frac{p_t({p_t}^2+\gamma^2)}{\gamma^2.{(p_t+\gamma)}^2} \geq 0.
\end{align*}
Using these results we conclude that $f(\gamma) \geq 0 $ $\forall$
$1 \geq p_t,\gamma \geq 0$, proving our lemma.

\end{proof}
\bibliographystyle{IEEEtran}
\bibliography{refs}
\end{document}